\documentclass{article}

\usepackage{amsmath}
\usepackage{arxiv}
\usepackage{diagbox} 
\usepackage{bm}
\usepackage[utf8]{inputenc} 
\usepackage[T1]{fontenc}    
\usepackage{url}            
\usepackage{booktabs}       
\usepackage{amsfonts}  
\usepackage{natbib} 
\usepackage{subfigure} 
\usepackage{graphicx}
\newcommand{\diag}{\mathrm{Diag}}
\newcommand{\Tr}{\mathrm{Tr}}

\usepackage{amsmath,amssymb,mathrsfs,amsthm}
\newtheorem{proposition}{Proposition}[]
\newtheorem{assumption}{Assumption}[]
\newtheorem{theoreme}{Theorem}[]
\newtheorem{remark}{Remark}[]
\newtheorem{lemma}{Lemma}[]
\newtheorem{corollary}{Corollary}[]
\usepackage{color}
\graphicspath{ {./images/} }
\usepackage{dsfont}

\usepackage{amsmath,amssymb,mathrsfs,amsthm}

\title{Holdout cross-validation for large non-Gaussian covariance matrix estimation using Weingarten calculus}

\author{
    Lamia Lamrani$^{1}$, Benoît Collins$^{2}$ and Jean-Philippe Bouchaud$^{3,4,5}$\\ 
    $^{1}$
    Laboratoire de Math\'ematiques et Informatique pour la Complexit\'e et les Syst\`emes,\\ Universit\'e Paris-Saclay, CentraleSup\'elec, 91192 Gif-sur-Yvette, France \\
    $^{2}$ 
  Department of Mathematics, Kyoto University, Kitashirakawa Oiwake-cho, Sakyo-ku, 
Kyoto 606-8502, Japan\\
$^{3}$Chair of Econophysics and Complex Systems, Ecole Polytechnique, 91128 Palaiseau, France\\
$^{4}$Capital Fund Management, 75007 Paris, France \\
$^{5}$Académie des Sciences, Paris 75006, France
}

\usepackage{amsthm}

\theoremstyle{definition}
\newtheorem{definition}{Definition}[section]

\begin{document}

\maketitle
\begin{abstract}
Cross-validation is one of the most widely used methods for model selection and evaluation; its efficiency for large covariance matrix estimation appears robust in practice, but little is known about the theoretical behavior of its error. In this paper, we derive the expected Frobenius error of the holdout method, a particular cross-validation procedure that involves a single train and test split, for a generic rotationally invariant multiplicative noise model, therefore extending previous results to non-Gaussian data distributions. Our approach involves using the Weingarten calculus and the Ledoit-Péché formula to derive the oracle eigenvalues in the high-dimensional limit. When the population covariance matrix follows an inverse Wishart distribution, we approximate the expected holdout error, first with a linear shrinkage, then with a quadratic shrinkage to approximate the oracle eigenvalues. Under the linear approximation, we find that the optimal train-test split ratio is proportional to the square root of the matrix dimension. Then we compute Monte Carlo simulations of the holdout error for different distributions of the norm of the noise, such as the Gaussian, Student, and Laplace distributions and observe that the quadratic approximation yields a substantial improvement, especially around the optimal train-test split ratio. We also observe that a higher fourth-order moment of the Euclidean norm of the noise vector sharpens the holdout error curve near the optimal split and lowers the ideal train-test ratio, making the choice of the train-test ratio more important when performing the holdout method. 

\end{abstract}


\section{Introduction}

The estimation of covariance matrices is as the heart of many methods in multivariate statistics with a broad range of applications such as physics \cite{taylor2013putting,schlitter1993estimation}, signal processing \cite{ottersten1998covariance,aubry2017geometric}, neurosciences \cite{varoquaux2010brain,zhuang2020technical} and finance \cite{laloux1999noise,froot1989consistent}. However, in a high dimension setting, when the number of features is not negligible with respect to the number of observations, the sample covariance matrix is inefficient \cite{ledoit2003honey}. Moreover, several applications such as Markowitz portfolio selection in finance additionally require the estimation of the inverse of the covariance matrix, for which the sample covariance is also largely inefficient \cite{michaud1989markowitz}. 

To address the lack of robustness of the sample covariance estimator in the high dimension setting, many methods have been developed, some of them which can be linked or understood with random matrix theory: eigenvalue clipping \cite{laloux1999noise, bun2017cleaning}, linear shrinkage \cite{ledoit2003honey}, more recently the Non Linear Shrinkage estimator (NLS) which have been first derived by Ledoit and Péché \cite{ledoit2011eigenvectors} and then further developed by several authors such as Ledoit and Wolf\cite{ledoit2012nonlinear,ledoit2020analytical} and Bouchaud, Bun and Potters \cite{bun2016rotational,bun2017cleaning,bun2018overlaps}, cross-validation (CV)\cite{bartz2016cross,berrar2019cross}, factor models \cite{fan2008high} and Bayesian estimation \cite{yang1994estimation} to quote the most common methods. In this paper, we will focus on the CV and more precisely on the holdout, which is a particular CV procedure.

CV is one of the most widely used methods for model selection and evaluation and finds applications in many areas such as ecology \cite{hijmans2012cross,yates2023cross}, physics \cite{iwama1989phillips}, medicine \cite{bradshaw2023guide}, and finance \cite{tan2020estimation}. Generally, all CV procedures involve splitting the data between a train and a test set \cite{berrar2019cross}. The most famous CV procedures are the k-folds CV, the leave-one-out CV (LOOCV) being a particular case of the former, and the holdout, also called validation or simple cross-validation \cite{berrar2019cross}. The k-folds procedure works as follows: the data are partitioned into $k$ equally sized test sets; then the model is evaluated on the test set and trained on the remaining data. Finally, the result is averaged among all the test sets. The LOOCV is a particular case of the k-fold procedure, where the $t$ data points are spitted into $t$ test sets, each of one observation. The holdout procedure consists of exactly one split between train and test without the averaging procedure, which is why it is sometimes called simple CV. To define the split between train and test to perform the holdout, we define the train-test ratio $k$ as the ratio between the number of data points and the length of the test set.

In the framework of covariance cleaning, the most common CV procedures are also the k-folds, the LOOCV and the holdout. A common recommendation to calibrate the k-fold CV is to choose a k-fold between $5$ and $10$ \cite{bartz2016cross,bun2018overlaps}. To our knowledge, the question of the CV error for large covariance matrix estimation has been addressed theoretically in the case of the holdout method. In an article from 2016 \cite{Lam2016,lam2015supplement}, Lam defines assumptions under which the Frobenius and Stein's loss of the holdout estimator, which he calls "Nonparametric Eigenvalue-Regularized Covariance Matrix Estimator" (NERCOME), converge, in the high-dimension limit, towards the error of the NLS of Ledoit and Wolf \cite{ledoit2011eigenvectors,ledoit2012nonlinear}. From this result, he derives an optimal region for the split of the holdout to ensure asymptotic convergence to the NLS. A few authors have raised concerns about whether this asymptotic property of the optimal split was really useful in the large but finite setting \citep{engel2017overview,tan2025estimation}. In \cite{lam2024holdout}, the authors computed the expected Frobenius error between the holdout estimator and the population covariance matrix with Gaussian data as a function of the oracle eigenvalues. Then assuming an inverse Wishart population covariance matrix, they obtained an explicit formula for the expected error in the high dimension limit and derived the optimal train-test ratio $k$, which is equal to the ratio between the number of data points and the length of the test set, and find it to be proportional to the square root of the matrix dimension. Monte Carlo analysis showed an excellent agreement between the expected error and the realized one, showing that the optimal splits derived in the high-dimensional limit are not equivalent in the finite but large setting. In this paper, we use the Weingarten calculus to tackle rotationally invariant multiplicative noises and also consider settings that do not satisfy the assumptions of Ledoit-Péché and Lam\cite{ledoit2011eigenvectors,Lam2016}.

Weingarten calculus can be seen as a generalization of the Wick theorem \cite{wick1950evaluation} for compact groups \cite{collins2017weingarten} and was first introduced by Don Weingarten in \cite{weingarten1978asymptotic}. Indeed, while Wick theorem makes it possible to compute joint moments of Gaussian variables, Weingarten calculus relies on the existence of a left-invariant Haar measure and a right-invariant measure on compact groups to compute joint moments \cite{collins2022moment}. In the framework of random matrix theory, the natural applications are on the groups of unitary and orthogonal matrices and the theory allows to compute moments of square random matrices that are rotationally invariant by conjugation by orthogonal or unitary matrices. Let us mention that there are also results to compute polynomial moments of rectangular matrices that are left- or right-invariant \cite{saad2013random}. Moreover, this theory finds many applications and has already been applied to the estimation of large covariance matrices and to finance to estimate the risk of the mean variance optimal portfolio \cite{collins2014integration, saad2013random}.

This paper is organized as follows, in Section \ref{Section 2} we set the definitions, in Section \ref{Section 3} we derive the expected Frobenius error of the holdout estimator for rotationally invariant multiplicative noise using the Weingarten calculus. Then we derive a linear and a quadratic approximation of the error for an inverse Wishart population covariance matrix and compute explicitly the optimal train-test ratio that minimizes the holdout error in the linear case. Finally, in Section \ref{Section 4} we verify the relevance of our approach on different distributions of the norm of the noise vector, namely  Gaussian, Student and Laplace.

\section{Definition}
\label{Section 2}
\subsection{Large covariance matrix cleaning}
Let us first define the multiplicative noise model that we will consider and its associated sample covariance matrix.
\begin{definition}
     Let $\bm{x}_{1},\dots,\bm{x}_{t}$ i.i.d. observations of the random vector $\bm{\xi}\in\mathbb{R}^{n}$ such that 
    \begin{equation}
        \mathbb{E}(\bm{\xi}\bm{\xi}^{T})=\bm{\mathds{1}}.
    \end{equation}
    Let us denote by $\bm{X}=(\bm{x}_{1}\vert\dots\vert\bm{x}_{t})\in\mathbb{R}^{n\times t}$ the noise matrix. Let us also introduce $\bm{\Sigma}\in\mathbb{R}^{n\times n}$, a semi-definite positive symmetric matrix and $\sqrt{\bm{\Sigma}}$ its symmetric square-root. Then the sample covariance matrix over the whole data set $\bm{E}$ equals

    \begin{equation}
        \bm{E}=\frac{1}{t}\sqrt{\bm{\Sigma}}\bm{X}\bm{X}^{T}\sqrt{\bm{\Sigma}}.
    \end{equation}
\end{definition}

$\bm{E}$ is an unbiased estimator of population covariance $\bm{\Sigma}$ that corresponds to the maximum likelihood estimator when the noise vectors are Gaussian. However, when $n$ is large without being negligible with respect to $t$ this estimator lacks efficiency. 

\subsubsection{Linear shrinkage}
A first alternative method to sample covariance to estimate the population covariance matrix is to linearly shrink the sample covariance matrix towards a target matrix while minimizing a distance to the population covariance matrix $\bm{\Sigma}$. Generally, the distance considered to measure the error is the Frobenius norm, and we will take the identity as our target matrix.

\begin{definition}
    The normalized Frobenius estimation error between an estimator $\bm{\Xi}$ and a population matrix $\bm{\Sigma}$ is \cite{potters2020first}
    \begin{equation}
    \Vert \bm{\Xi} - \bm{\Sigma}\Vert^{2}_{\textrm{F}}=\tau ((\bm{\Xi} - \bm{\Sigma})^{2}) ,
    \end{equation}
    where $\tau(\bullet)=\frac{1}{n}\Tr(\bullet)$ is the normalized trace. The subscript $F$ will be dropped henceforth as we will only consider the Frobenius norm throughout this paper.
\end{definition}

Then the linear shrinkage problem amounts to
\begin{equation}
\label{eq:linear_shrinkage_formulation}
    \underset{r\in[0,1]}{\mathrm{argmin}}\Vert r\bm{E}+(1-r)\mathds{1} - \bm{\Sigma} \Vert_{F}^{2} 
\end{equation}

and the solution to this optimization problem is provided in the following proposition.

\begin{proposition}
\label{prop:linear_shrinkage}
Let us consider $\bm{\Sigma}$ a semidefinite positive matrix and $\bm{E}$ the sample covariance estimator of $\bm{\Sigma}$. We assume that both $\bm{\Sigma}$ and $\bm{E}$ are normalized such that $\mathbb{E}\left[\tau(\bm{\Sigma})\right]=\mathbb{E}\left[\tau(\bm{E})\right]$=1. Then the optimal linear shrinkage coefficient that minimizes the expected Frobenius error of estimation written in Eq.\ \eqref{eq:linear_shrinkage_formulation} is \cite{ledoit2003honey}
    \begin{equation}
        r=\frac{\mathbb{E}\left[\tau(\bm{\Sigma})^{2} \right]-1}{\mathbb{E}\left[\tau(\bm{E})^{2} \right]-1}.
    \end{equation}
\end{proposition}

\begin{proof}

Let us first expand the Frobenius error of the linear shrinkage as a function of the moments of $\bm{E}$ and $\bm{\Sigma}$.
    \begin{equation}
    \label{eq:frob_error_lin_shrink_formulation}
    \Vert r\bm{E}+(1-r)\bm{\mathds{1}}-\bm{\Sigma}\Vert_{F}^{2}
    =r^{2}\tau(\bm{E}^{2}) + (1-r)^{2} + \tau(\bm{\Sigma}^{2}) + 2r(1-r)\tau(\bm{E})-2r\tau(\bm{E}\bm{\Sigma}).
    \end{equation}

    Then using the assumptions  
    \begin{equation}
        \mathbb{E}\left[\tau(\bm{\Sigma)} \right]=\mathbb{E}\left[\tau(\bm{E)} \right]=1,
    \end{equation}
    and using the following equality, valid in the high dimensional limit, 
    \begin{equation}
        \mathbb{E}\left[\tau(\bm{E\Sigma)} \right]=\mathbb{E}\left[\tau(\bm{\Sigma}^{2})\right],
    \end{equation}
the expected Frobenius error from Eq.\ \eqref{eq:frob_error_lin_shrink_formulation} simplifies to 
\begin{equation}
 \mathbb{E}\left[(r\bm{E}+(1-r)\mathds{1}-\bm{\Sigma})^{2} \right] = r^{2}(\mathbb{E}\left[\tau(\bm{E}^{2})\right]-1)-2r(\mathbb{E}\left[\tau(\bm{\Sigma}^{2}) \right]-1) +1 + \mathbb{E}\left[\tau(\bm{\Sigma}^{2})\right]. 
\end{equation}

Then by cancelling the derivative of the expected Frobenius error, the $r_{\textrm{opt}}$ that minimizes Eq.\ \eqref{eq:frob_error_lin_shrink_formulation} is equal to 
\begin{equation}
    r_{\textrm{opt}}=\frac{\mathbb{E}\left[\tau(\bm{\Sigma})^{2} \right]-1}{\mathbb{E}\left[\tau(\bm{E})^{2} \right]-1}.
\end{equation}
\end{proof}

\subsubsection{Oracle estimator}
Although the linear shrinkage already achieves very satisfactory performances for large covariance matrix estimation in terms of minimizing the Frobenius error with the population covariance matrix \cite{ledoit2003honey}, it has been argued that the shrinkage of large covariance matrices should be treated in essence as a nonlinear problem \cite{ledoit2011eigenvectors}. 

The NLS consists in defining a rotationally invariant estimator (RIE) that minimizes the Frobenius error with the population covariance matrix. Let us first define the notion of a RIE.

\begin{definition}
$\bm{\Xi}(\bm{E})$ is a rotationally invariant estimator (RIE) of the population covariance matrix $\bm{\Sigma}$, where $\bm{E}$ denotes the sample covariance estimator. This estimator satisfies the following distributional equality for all orthogonal matrices $\bm{O}$\cite{bun2016rotational}
\begin{equation}
 \bm{\Xi}(\bm{O}\bm{E}\bm{O}^{T})\overset{\textrm{distr}}{=}\bm{O}\Xi(\bm{E}) \bm{O}^{T}.   
\end{equation}
\end{definition}

The introduction of an RIE comes from the idea that when no information is known about the eigenvector basis of the population covariance matrix, one can consider that all eigenvector bases are equiprobable and therefore keep the sample eigenvectors.

Then, the RIE which minimizes the Frobenius error with population covariance $\bm{\Sigma}$, commonly called the oracle estimator in the literature, equals \cite{ledoit2011eigenvectors}

\begin{equation}
    \bm{\Xi}^{O}(\bm{E},\bm{\Sigma})=\bm{V}\diag(\bm{V}^{T}\bm{\Sigma V})\bm{V}^{T},
\end{equation}

where $\bm{V}$ are the sample eigenvectors of $\bm{E}$ and $\diag(\bullet)$ is the operator that sets the out-of-diagonal entries to zero.

This estimator still depends on the unknown quantity $\bm{\Sigma}$, however, under some regularity properties listed in the following theorem, Ledoit and Péché were able to derive the oracle eigenvalues in the high-dimensional limit \cite{ledoit2011eigenvectors}.

\begin{theoreme}
Let us call $\bm{E}$ the sample covariance of size $n\times n$ generated by $t$ i.i.d. observations with covariance population matrix $\bm{\Sigma}$. Let us assume that the noise used to generate $\bm{E}$ is independent of $\bm{\Sigma}$, has a finite 12th absolute central moment, and that the distribution of $\Sigma$ converges for $n$ tends to $\infty$ at every point of continuity to a nonrandom distribution which defines a probability measure and whose support is included in a compact interval.  Let $q=\frac{n}{t}$ and $g_{\bm{E}}$ the Stieltjes transform of $\bm{E}$. For $\lambda$, an eigenvalue of $\bm{E}$, let $\xi_{\lambda}(\bm{E})$ be the corresponding eigenvalue of the oracle estimator. Then, in the high-dimensional limit,  for $z=\lambda+i\eta$ with $\lambda\in Sp(\bm{E})$ and $\eta>0$ \cite{ledoit2011eigenvectors}
\begin{equation}
\label{ledoit_peche}
\xi_{\lambda}(\bm{E})=\lim\limits_{\eta \rightarrow 0^{+}} \frac{\lambda}{\vert 1-q+q\lambda g_{\bm{E}}(\lambda+i\eta)\vert^2}.
\end{equation}
\begin{flushright}
(Ledoit-P\'ech\'e formula)
\end{flushright}
\end{theoreme}

Let us mention that although this formula is exact only in the high dimension limit, i.e. $n,t\to\infty$ and $\frac{n}{t}\to q>0$, it remains a good approximation for large but finite dimensions \cite{ledoit2012nonlinear}.

Moreover, let us recall that when $\bm{\Sigma}$ is an inverse Wishart covariance matrix and the data are Gaussian, the oracle estimator is equal in the high dimension limit to the optimal linear shrinkage defined in Proposition\ \eqref{prop:linear_shrinkage}. Indeed, with $\lfloor x \rfloor$ being the integer part of $x$,
\begin{definition}
\label{def:inverse_wishart}
    let $n\in\mathbb{N}^{*}$ and $p>0$. Let us define $q^{*}=\frac{p}{1+p}$ and $t=\lfloor \frac{n}{q^{*}}\rfloor$. Consider $(\bm{m_{i}})_{1\leq i\leq t}$ i.i.d. vectors of $\mathbb{R}^{n}$ such that $\bm{m_{i}}\sim\mathcal{N}(0,\bm{C})$. Let us define $\bm{M}=(\bm{m_{1}},\dots,\bm{m_{t}})\in\mathbb{R}^{n\times t}$. 
Then the following matrix is Wishart $(n,p)$
\begin{equation}
\bm{W}=\frac{1}{t}\bm{MM}^{T},
\end{equation}

and an inverse Wishart $(n,p)$ is then
\begin{equation}
\bm{\mathcal{W}}^{-1}_{np}=(1-q^{*})\bm{W}^{-1}.
\end{equation}
\end{definition}
In this paper, we will only consider the situation $\bm{C}=\bm{\mathds{1}}$, which corresponds to the white inverse Wishart case.

\begin{proposition}
\label{prop:oraclle_linear_shrinkage}
When the population matrix is a white inverse Wishart and the data are Gaussian, the  oracle estimator recovers the optimal linear shrinkage in the Bayesian sense in the high dimension limit \cite{potters2020first}
\begin{equation}
\label{eq:linear_shrink}
\bm{\Xi}^{\textrm{O}}=r\bm{E}+(1-r)\bm{\mathds{1}} \hspace{4mm} \mbox{, with:} \hspace{4mm} r=\frac{\mathbb{E}\left[\tau(\bm{\Sigma}^{2})\right]-1}{\mathbb{E}\left[\tau(\bm{E}^{2})\right]-1}.
\end{equation}
\end{proposition}

\subsubsection{Holdout and CV estimators}
A third remedy to the weaknesses of the sample covariance matrix for the estimation of the large covariance matrix is the CV estimator family. Let us introduce the holdout estimator, which will be the object of our study, and the k-fold, which is the most widely used CV procedure for covariance estimation \cite{berrar2019cross,lam2024holdout}.

\begin{definition}
\label{def:holdout_estimator}
   Let us consider $X=\{\bm{x}_{1}, \dots, \bm{x}_{t}\}$ an i.i.d. set of centered noise observations and $\bm{\Sigma}$ the population covariance.
    Let $t_{\textrm{out}}$ an integer between $1$ and $t-1$, $k= \frac{t}{t_{\textrm{out}}}$, $t_{\textrm{in}}=t-t_{\textrm{out}}$, $q_{\textrm{in}}=\frac{n}{t_{\textrm{in}}}$ and $q_{\textrm{out}}=\frac{n}{t_{\textrm{out}}}$. Let us call $(\mathcal{I}_{\textrm{in}},\mathcal{I}_{\textrm{out}})$ a partition of $[1,t]$ such that $\mathcal{I}_{\textrm{in}}$ is of size $t_{\textrm{in}}$ and $\mathcal{I}_{\textrm{out}}$ of size $t_{\textrm{out}}$.  Let us call $\bm{E}_{\textrm{in}}$ the sample covariance over the noise observations $\bm{x}_{\mathcal{I}_{\textrm{in}}(1)},\dots,\bm{x}_{\mathcal{I}_{\textrm{in}}(t_{\textrm{in}})}$ (train set) and $\bm{E}_{\textrm{out}}$ the sample covariance over the noise observations $\bm{x}_{\mathcal{I}_{\textrm{out}}(1)}, \dots, \bm{x}_{\mathcal{I}_{\textrm{out}}(t_{out})}$ (test set). Let us call $\bm{V}_{\textrm{in}}$ the orthogonal matrix containing the eigenvectors of $\bm{E}_{\textrm{in}}$. Then the holdout estimator is defined as \cite{Lam2016,lam2015supplement}
    
\begin{equation}    \bm{\Xi}^H(\bm{V}_{\textrm{in}},\bm{E}_{\textrm{out}}) := \bm{\Xi}^O(\bm{E},\bm{E}_{\textrm{out}})=\bm{V}_{\textrm{in}}\diag(\bm{V}_{\textrm{in}}^{T}\bm{E}_{\textrm{out}}\bm{V}_{\textrm{in}})\bm{V}_{\textrm{in}}^{T}.
    \end{equation}

    \begin{flushright}
        (holdout estimator)
    \end{flushright}
\end{definition}

Let us also define the k-fold partition.

\begin{definition}\label{def:kpartition}
    Let $t $ be the total number of observations, and let $t_{\textrm{out}}$ be an integer between $1$ and $t-1$ such that $t_{\textrm{out}}$ divides $t$. Define $k = t/t_{\textrm{out}}$ and $t_{\textrm{in}} = t - t_{\textrm{out}}$. We partition the set of indices $\mathcal{I} = \{1, 2, \dots, t\}$ into $k$ disjoint subsets $\{\mathcal{I}_{\textrm{out}, l}\}_{l=1}^{k}$, each of size $t_{\textrm{out}}$.Then, for each $l$, the training set indices are defined as $\mathcal{I}_{\textrm{in}, l} = \mathcal{I} \setminus \mathcal{I}_{\textrm{out}, l}$.
\begin{flushright}
($k$-fold index partition)
\end{flushright}

\end{definition}

\begin{definition}
    Consider $X=\{\bm{x}_{1}, \dots, \bm{x}_{t}\}$ an i.i.d. set of centered noise observations and $\bm{\Sigma}$ the population covariance of the data generated with the noise matrix $\bm{X}$. Let $t_\textrm{out}$ be an integer between $1$ and $t-1$ such that $t_\textrm{out}$ divides $t$, and let $\{\mathcal{I}_{\textrm{out}, l}, \mathcal{I}_{\textrm{in}, l}\}$ be a $k$-fold index partition (definition\ ~\eqref{def:kpartition}). For each $1 \leq l \leq k$, let $\bm{E}_{\textrm{in}, l}$ be the sample covariance of $\bm{\Sigma}$ over the noise observations $\bm{X}_{\textrm{in}, l} = \{ \bm{x}_i | i \in \mathcal{I}_{\textrm{in}, l} \}$ (train set) and $\bm{E}_{\textrm{out}, l}$ the sample covariance over s $\bm{X}_{\textrm{out}, l} = \{ \bm{x}_i | i \in \mathcal{I}_{\textrm{out}, l} \}$ (test set). Let $\bm{V}_{\textrm{in}, l}$ the eigenvectors of $\bm{E}_{\textrm{in}, l}$ ranked by increasing eigenvalues. Then the $k$-fold CV estimator is \cite{Lam2016,lam2015supplement}
\begin{equation}
\label{k_fold_cv_def}
\bm{\Xi}^\textrm{CV} = \frac{1}{k} \sum_{l=1}^{k} \bm{\Xi}_l^\textrm{H} = \frac{1}{k} \sum_{l=1}^k \bm{V}_{\textrm{in}, l}\diag(\bm{V}_{\textrm{in}, l}^{T}\bm{E}_{\textrm{out}, l}\bm{V}_{\textrm{in}, l})\bm{V}_{\textrm{in}, l}^{T}.
\end{equation}

\begin{flushright}
    ($k$-fold CV estimator)
\end{flushright}

\end{definition}

Let us now state the convergence result of Lam with the associated assumptions \cite{Lam2016,lam2015supplement}.

\begin{assumption}\label{assumptions_lam}\item[(H1)] The observations can be written $\bm{X}=\sqrt{\bm{\Sigma}}\bm{Y}$ where $\bm{Y}$ is a $n\times t$ matrix of real or complex random variables with zero mean and unit variance such that $\mathbb{E}\left(\vert Y_{ij} \vert^{l} \right)\leq B<\infty$ for some constant $B$ and for $2<l\leq 20$. 
        \item[(H2)] The population covariance matrix $\bm{\Sigma}$ is non-random and $\Vert \bm{\Sigma}\Vert_{L^{2}}=\mathcal{O}(1)$ and $\bm{\Sigma}\neq \sigma^{2}\bm{\mathds{1}}$
    \item [(H3)] $\frac{n}{t}\to q>0$ as $n\to\infty$

    \item [(H4)] Let us consider $(\tau_{1}\dots \tau_{n})$ the eigenvalues of $\bm{\Sigma}_{n}$, the empirical spectral distribution of $\bm{\Sigma}_{n}$
    \begin{equation}
        H_{n}(\tau)=\frac{1}{n}\sum_{j=1}^{n}\mathds{1}_{[\tau_{j},+\infty[}(\tau)
    \end{equation}
    converges to a non-random limit $H(\tau)$ at every point of continuity of $H$. Moreover, H defines a probability distribution function whose support $Supp(H)$ is included in a compact interval $[h_{1},h_{2}]$ where $0<h_{1}\leq h_{2}<+\infty$.    
\end{assumption}

\begin{theoreme}
    Under the assumptions listed in \ref{assumptions_lam}, the holdout estimator error converges towards the error of the NLS, with respect to the Frobenius error and Stein's loss provided that the split satisfies $\frac{t_{\textrm{in}}}{t}\xrightarrow[t \to \infty]{} 1$, $t_{\textrm{out}} \xrightarrow[t \to \infty]{} \infty$ and
\begin{equation}
\label{eq:lam_condition}
   \sum_{t\geq 1}\frac{n}{t_{\textrm{out}}^{5}}<\infty. 
\end{equation}
\end{theoreme}

The last two assumptions are included in the assumptions of the Ledoit and Péché's formula. For example, choosing $t_{\textrm{out}}$ proportional to the square root of $n$ would satisfy the conditions written in the previous theorem. However, the first assumption written in \ref{assumptions_lam} is restrictive on the moments of the noise. In this paper, we will consider, as an example, a student norm for the noise distribution for which neither the assumptions on the finite moments of Lam nor those of Ledoit and Péché's formula will be satisfied \cite{ledoit2011eigenvectors,Lam2016}.

\subsection{Orthogonal Weingarten calculus}
Weingarten calculus can be seen as a generalization of the Wick theorem to compute joint moments rotationally invariant quantities. Let us first recall some definitions and results, which can be found in \cite{collins2017weingarten}, that will be needed to study the holdout error. 

\begin{definition}
    Let $\bm{A}$ be a random square matrix of size $n$, $\bm{A}$ is said to be rotationally invariant if the following equality is verified in distribution for every $\bm{O}\in\mathcal{O}(n)$,  the group of orthogonal matrices of size $n$,
    \begin{equation}
        \bm{A}=\bm{O}\bm{A}\bm{O}^{T}.
    \end{equation}
\end{definition}

Although we will only focus on square matrices in this paper, let us mention that a similar notion exists for rectangular matrices called the left-right invariance. 

Let us note $S_{2k}$ the group of permutations in $\{1,2,\dots,2k\}$ and $M_{\textrm{2k}}$ the group of pair partitions, i.e. the group of partitions of $\{1,2,\dots,2k\}$ with each subset being of size two; each pair partition can be uniquely expressed in the form 
\begin{equation}
\{ \{m(1),m(2)\},\{m(3),m(4)\},...,\{m(2n-1),m(2n)\}\} 
\end{equation}
 with $m(2i-1)<m(2i)$ for $1\leq i\leq n$ and with $m(1)<m(2)<...<m(2n-1)$.

\begin{definition}
    For each $\sigma\in M_{\textrm{2k}}$, let us define its associated graph $\Gamma(\sigma)$ whose vertices are $1,2,\dots,2k$ and edge set
    \begin{equation}
        \left\{ \{2i-1,2i \}\vert i=1,2,..,k \right\} \hspace{2mm} \cup \hspace{2mm} \left\{ \{\sigma(2i-1),\sigma(2i) \}\vert i=1,2,..,k \right\}.
    \end{equation}
If two pairs $\{2i-1,2i \}$ and $\{\sigma(2i-1),\sigma(2i) \}$ were to coincide, two edges must be added to the graph.

Then the Gram matrix associated to the group $M_{\textrm{2k}}$, $G^{\mathcal{O}(n)}_{k}=(G^{\mathcal{O}(d)}(\sigma,\tau))_{\sigma,\tau\in M_{\textrm{2k}}}$ is given by
\begin{equation}
    G^{O(n)}(\sigma,\tau)=d^{\textrm{loop}(\sigma,\tau)},
\end{equation}
where $\textrm{loop}(\sigma,\tau)$ is the number of connected components of $\Gamma(m,n)$. 

The orthogonal Weingarten matrix of $M_{2k}$, $Wg_{k}^{\mathcal{O}(n)}=(Wg^{\mathcal{O}(n)}(\sigma,\tau))_{\sigma,\tau\in M_{2k}}$ is the pseudo inverse of the Gram matrix.
    
\end{definition}

\begin{definition}
    Let $\Gamma(\sigma)$ be the graph associated with the permutation $\sigma\in S_{2k}$. 
    In the connected component of the graph, if the number of vertices are 
    \begin{equation}
        2\mu_{1}\geq 2\mu_{2}\geq \dots\geq 2\mu_{l}
    \end{equation}
    then $\mu=(\mu_{1},\dots,\mu_{l})$ is referred to as the coset type of $\Gamma(\sigma)$; $l$ represents the number of connected components of $\Gamma(\sigma)$.
\end{definition}

\begin{definition}
    Let $\bm{A}\in\mathbb{R}^{n\times n}$ and $\sigma\in S_{2k}$ of coset type $(\mu_{1},\dots,\mu_{l})$, we define the following operator
    \begin{equation}
        \Tr_{\sigma}'(\bm{A})=\prod_{j=1}^{l}\Tr(\bm{A}^{\mu_{j}}).
    \end{equation}
\end{definition}

\begin{definition}
    For $\sigma\in S_{2k}$ and a $2k$-tuple $i=(i_{1},\dots,i_{2k})$ of positive integers, we define
    \begin{equation}
        \delta_{\sigma}'(i)=\prod_{s=1}^{k}\delta_{i_{\sigma(2s-1)}i_{\sigma(2s)}}.
    \end{equation}
If $\sigma\in M_{2k}$, then $\delta_{\sigma}$ simplifies to 
\begin{equation}
    \delta_{\sigma}'(i)=\prod_{\{a,b\}\in\sigma}\delta_{i_{a}i_{b}}
\end{equation}
where the products run over all the pairs contained in the pair permutation $\sigma$.
\end{definition}

\begin{theoreme}
Let $\bm{A}\in\mathbb{R}^{n\times n}$ a symmetric and rotationally invariant matrix. Then for any sequence $i=(i_{1},\dots,i_{2k})$ we have
\begin{equation}
\label{eq:orthogonal_weingarten_formula}
    \mathbb{E}\left( A_{i_{1}i_{2}}A_{i_{3}i_{4}}\dots A_{i_{2k-1}i_{2k}} \right) = \sum_{\sigma,\tau\in M_{2k}} \delta_{\sigma}'(i)Wg^{O}(\sigma^{-1}\tau;n)\mathbb{E}(\Tr_{\tau}'(\bm{A})).
\end{equation}
\end{theoreme}
This theorem allows us to access the expectation of joint local entries of any rotationally invariant squared random matrix by linking them with the expectation of global tracial moments.

\subsubsection{Computation of the orthogonal Weingarten matrix in $M_{4}$}

To derive the expected Frobenius error of estimation between the holdout estimator and the population covariance matrix, the Weingarten function on $M_{4}$ will be needed. Therefore, in preparation for the main result of this paper, let us derive the Weingarten matrix in $M_{4}$.

$M_{4}$ includes three pair-partitions: $(12)(34)$,$(13)(24)$ and $(14)(23)$.

\begin{table}
\centering
\begin{equation}
\begin{tabular}{c|ccc}
    \diagbox{$\sigma$}{$\tau$}& $(12)(34)$ & $(13)(24)$ & $(14)(23)$ \\ \hline
    $(12)(34)$ & $n^{2}$ & $n$ & $n$ \\
    $(13)(24)$ & $n$ & $n^{2}$ & $n$ \\
    $(14)(23)$ & $n$ & $n$ & $n^{2}$ \\
\end{tabular}
\end{equation}
\caption{Gram matrix of $M_{4}$} 
\label{tab:gram} 
\end{table}

Then the Weingarten matrix in $M_{4}$ is equal to the inverse of the Gram matrix written in table \ref{tab:gram} 
\begin{equation}
    W^{\textrm{O}}=\begin{bmatrix}
n^{2} & n & n \\
n & n^{2}  & n   \\
n & n & n^{2} 
\end{bmatrix}^{-1},
\end{equation}

and by inverting it we obtain
\begin{equation}
\label{eq:weingarten_matrix}
    W^{\textrm{O}}=\begin{bmatrix}
\alpha & \beta & \beta \\
\beta & \alpha & \beta  \\
\beta & \beta & \alpha
\end{bmatrix}
\end{equation}
where 
\begin{equation}
    \alpha=\frac{n+1}{n(n-1)(n+2)}
\end{equation}
and 
\begin{equation}
    \beta=-\frac{1}{n(n-1)(n+2)}.
\end{equation}

The inversion of large formal Weingarten matrices can become a computational challenge. A table of values of the orthogonal Weingarten function up to $M_{12}$ has been computed in \cite{collins2009some}.

\section{Frobenius error of the holdout method and optimal split}
\label{Section 3}

In this section, we compute the expected Frobenius error of the holdout estimator with respect to the population covariance matrix in the high dimension limit. We consider data generated by a rotationally invariant multiplicative noise model, defined in section \ref{Section 2}. The first step of the proof involves the use of Weingarten calculus and more precisely of the Weingarten matrix computed in Eq. \eqref{eq:weingarten_matrix} to link the moments of the out-of-sample covariance matrix with those of the population covariance matrix and of the noise matrix. Then, as in the holdout error proof in the Gaussian case \cite{lam2024holdout}, the error is written as a function of the oracle eigenvalues which can be computed in the high-dimension limit using the Ledoit-Péché formula \cite{ledoit2011eigenvectors} .
\subsection{General holdout error}

\begin{proposition}
\label{prop:general_holdout_error}
Let us consider $\bm{\Sigma}$ a positive-defined matrix in $\mathbb{R}^{n \times n}$ and $(\bm{x}_{\textrm{i}})_{1\leq i\leq t}$ i.i.d. observations of the rotationally invariant random vector $\bm{\xi}\in\mathbb{R}^{n}$ such that
\begin{equation}
\label{eq:noise_assumption}
    \mathbb{E}(\bm{\xi}\bm{\xi}^{T})=\mathds{1},
\end{equation}
and 
\begin{equation}
    \mathbb{E}(\Vert\bm{\xi}\Vert^{4})<\infty,
\end{equation}

where $\Vert \bullet \Vert$ is the Euclidean norm.
 
Let $\bm{\Xi}^\textrm{H}$ the holdout estimator associated with the index partition $\{\mathcal{I}_\textrm{in},\mathcal{I}_\textrm{out}\}$ defined from $t_\textrm{out}$ integer between $1$ and $t-1$, with $\bm{V}_\textrm{in}$ the eigenvector of the sample covariance $\bm{E}_\textrm{in}$ on the train partition $\mathcal{I}_\textrm{in}$.
 
 Then the expected Frobenius error between the holdout estimator and the population covariance $\bm{\Sigma}$ is in the high dimension limit
    \begin{equation}
    \label{eq:general_holdout_error_formula}
        \mathbb{E}\left(\Vert \bm{\Xi}^{H} - \bm{\Sigma} \Vert^{2}\right) = \left[ \frac{k}{t}\left(\frac{3 \mathbb{E}\left[ \Tr((\bm{\xi}\bm{\xi}^{T})^{2})\right]}{n(n+2)} - 1\right) - 1 \right] \mathbb{E}\left[\tau(\diag(\bm{V}_{\textrm{in}}^{t}\bm{\Sigma} \bm{V}_{\textrm{in}})^{2})\right]+ \mathbb{E}\left[\tau(\bm{\Sigma}^{2})\right].
    \end{equation}
\end{proposition}

\begin{remark}
    Although Eq. \eqref{eq:general_holdout_error_formula} is valid only in the high-dimension limit, that is, when $n,t\to \infty$ with $\frac{n}{t}=q>0$ fixed. It should be noted that $k$, a real number between $1$ and $t$, is kept in the equation. Indeed, an interesting application of deriving the holdout error is being able to compute the optimal $k$ which minimizes it. 
\end{remark}

\begin{proof}
The expected Frobenius error of estimation between the holdout estimator and the population covariance matrix can be expanded as followed
\begin{equation}
    \label{eq:frob_error_expansion}
    \mathbb{E}\left[\tau((\bm{\Xi}^{\textrm{H}}- \bm{\Sigma})^{2}) \right]= \mathbb{E}\left[\tau((\bm{\Xi}^{\textrm{H}})^{2}) \right]- 2\mathbb{E}\left[ \tau(\bm{\Xi}^{\textrm{H}}\bm{\Sigma})\right]+ \mathbb{E}\left[\tau(\bm{\Sigma}^{2}) \right].
\end{equation}
Let us firstly focus on the first term of Eq.\ \eqref{eq:frob_error_expansion}

\begin{align}
    \mathbb{E}\left[\tau((\bm{\Xi}^{\textrm{H}})^2)\right]&=\mathbb{E}\left[\tau(\diag(\bm{V}_{\textrm{in}}^{t}\bm{E}_{\textrm{out}}\bm{V_{\textrm{in}}})^{2})\right]\\
    &=\frac{1}{n}\sum_{i,j,k,l,m} \mathbb{E}\left[(\bm{V}_{\textrm{in}})_{ji}(\bm{V}_{\textrm{in}})_{ki}(\bm{V}_{\textrm{in}})_{li}(\bm{V}_{\textrm{in}})_{mi}\right]\mathbb{E}\left[(\bm{E}_\textrm{out})_{jk}(\bm{E}_{\textrm{out}})_{lm}\right],
\end{align}
where the last equality has been obtained using the independence between the train and test data.

 Let us remind that the sample covariance $\bm{E}_{\textrm{out}}$ can be expressed as a function of the population covariance $\bm{\Sigma}$ and of the noise matrix $\bm{X}_\textrm{out}\in\mathbb{R}^{n\times t_{\textrm{out}}}$. To simplify the notations, $\bm{X}_\textrm{out}\in\mathbb{R}^{n\times t_{\textrm{out}}}$ will be noted as $\bm{X}$ and $\bm{E}_{\textrm{out}}$ can be written
\begin{equation}
    \bm{E}_{\textrm{out}}=\frac{1}{t_{out}}\sqrt{\bm{\Sigma}}\bm{X}\bm{X}^{T}\sqrt{\bm{\Sigma}},
\end{equation}

where we assume the noise matrix $\bm{X}\bm{X}^{T}$ to be asymptotically free from the population covariance $\bm{\Sigma}$.

To lighten the notations, let us define 
\begin{equation}
    \sqrt{\bm{\Sigma}} =\Tilde{\bm{\Sigma}}, 
\end{equation}

and let us focus on computing $\mathbb{E}\left[(\bm{E}_\textrm{out})_{jk}(\bm{E}_{\textrm{out}})_{lm}\right]$.

\begin{align*}
    \mathbb{E}\left[(E_\textrm{out})_{jk}(E_{\textrm{out}})_{lm}\right]&=\frac{1}{(t_{\textrm{out}})^{2}}\mathbb{E}\left(\sum_{i_{1},i_{2},i_{3},i_{4}=1}^{n}\sum_{t_{1},t_{2}=1}^{t_{\textrm{out}}} \Tilde{\Sigma}_{ji_{1}}X_{i_{1}t_{1}}X_{i_{2}t_{1}} \Tilde{\Sigma}_{i_{2}k} \Tilde{\Sigma}_{li_{3}}X_{i_{3}t_{2}}X_{i_{4}t_{2}} \Tilde{\Sigma}_{i_{4}m} \right) \\
    &= \frac{1}{(t_{\textrm{out}})^{2}}\sum_{i_{1},i_{2},i_{3},i_{4}=1}^{n}\sum_{t_{1},t_{2}=1}^{t_{\textrm{out}}} \mathbb{E} \left( \Tilde{\Sigma}_{ji_{1}}\Tilde{\Sigma}_{i_{2}k} \Tilde{\Sigma}_{li_{3}}\Tilde{\Sigma}_{i_{4}m}\right) \mathbb{E}\left( X_{i_{1}t_{1}}X_{i_{2}t_{1}} X_{i_{3}t_{2}}X_{i_{4}t_{2}} \right)\\
    &= \frac{1}{(t_{\textrm{out}})^{2}}\sum_{i_{1},i_{2},i_{3},i_{4}=1}^{n}\sum_{t_{1}\neq t_{2}}^{t_{\textrm{out}}} \mathbb{E} \left( \Tilde{\Sigma}_{ji_{1}}\Tilde{\Sigma}_{i_{2}k} \Tilde{\Sigma}_{li_{3}}\Tilde{\Sigma}_{i_{4}m}\right) \mathbb{E}\left( X_{i_{1}t_{1}}X_{i_{2}t_{1} }\right) \mathbb{E}\left(X_{i_{3}t_{2}}X_{i_{4}t_{2}} \right)\\
    &+ \frac{1}{t_{\textrm{out}}}  \sum_{i_{1},i_{2},i_{3},i_{4}=1}^{n} \mathbb{E} \left( \Tilde{\Sigma}_{ji_{1}}\Tilde{\Sigma}_{i_{2}k} \Tilde{\Sigma}_{li_{3}}\Tilde{\Sigma}_{i_{4}m}\right) \mathbb{E}\left( X_{i_{1}1}X_{i_{2}1} X_{i_{3}1}X_{i_{4}1} \right).
\end{align*}

The last equality was obtained using the fact that the columns of the noise matrix $\bm{X}$ are i.i.d.

Moreover by Eq.\ \eqref{eq:noise_assumption}, for $1\leq i_{1},i_{2}\leq n$ we have
\begin{equation}
   \mathbb{E}(X_{i_{1}t_{1}}X_{i_{2}t_{1}})=\delta_{i_{1}i_{2}}.
\end{equation}
Therefore,
\begin{align*}
    \mathbb{E}\left[(E_\textrm{out})_{jk}(E_{\textrm{out}})_{lm}\right]&=\frac{t_{\textrm{out}}-1}{t_{\textrm{out}}}\sum_{i_{1},i_{3}=1}^{n} \mathbb{E} \left( \Tilde{\Sigma}_{ji_{1}}\Tilde{\Sigma}_{i_{1}k} \Tilde{\Sigma}_{li_{3}}\Tilde{\Sigma}_{i_{3}m}\right) \\
    &+ \frac{1}{t_{\textrm{out}}} \sum_{i_{1},i_{3},i_{4},i_{6}=1}^{n} \mathbb{E} \left( \Tilde{\Sigma}_{ji_{1}}\Tilde{\Sigma}_{i_{2}k} \Tilde{\Sigma}_{li_{3}}\Tilde{\Sigma}_{i_{4}m}\right) \mathbb{E}\left( X_{i_{1}1}X_{i_{2}1} X_{i_{3}1}X_{i_{4}1} \right),
\end{align*}

which simplifies to
\begin{equation}  
\label{eq:prod_eout_before_weingarten}
   \mathbb{E}\left[(E_\textrm{out})_{jk}(E_{\textrm{out}})_{lm}\right] = \frac{t_{\textrm{out}}-1}{t_{\textrm{out}}}\Sigma_{jk}\Sigma_{lm} +  \frac{1}{t_{\textrm{out}}}  \sum_{i_{1},i_{3},i_{4},i_{6}=1}^{n} \mathbb{E} \left( \Tilde{\Sigma}_{ji_{1}}\Tilde{\Sigma}_{i_{2}k} \Tilde{\Sigma}_{li_{3}}\Tilde{\Sigma}_{i_{4}m}\right) \mathbb{E}\left( X_{i_{1}1}X_{i_{2}1} X_{i_{3}1}X_{i_{4}1} \right)
\end{equation}

Let us now use Weingarten calculus and more precisely the Weingarten matrix obtained in Eq.\ \eqref{eq:weingarten_matrix} to link the term $\mathbb{E}(X_{i_{1}1}X_{i_{2}1}X_{i_{3}1}X_{i_{4}1})$ to the expected value of the tracial moments of the noise matrix. Let us first note that we can re-write $\mathbb{E}(X_{i_{1}1}X_{i_{2}1}X_{i_{3}1}X_{i_{4}1})$ using the noise vector $\bm{\xi}=(\xi_{1},\dots \xi_{n})^{T}$ as 

\begin{equation}
    \mathbb{E}(X_{i_{1}1}X_{i_{2}1}X_{i_{3}1}X_{i_{4}1}) = \mathbb{E}\left[ \xi_{i_{1}}\xi_{i_{2}}\xi_{i_{3}}\xi_{i_{4}}\right],
\end{equation}
and let us prepare the components appearing in the formula written in Eq.\ \eqref{eq:orthogonal_weingarten_formula}
\begin{equation}
    \mathbb{E}\left[\Tr'_{\tau=(12)(34)}(\bm{\xi}\bm{\xi}^{T}) \right] = \mathbb{E}\left[ \left(\Tr(\bm{\xi}\bm{\xi}^{T})\right)^{2} \right].
\end{equation}

and 
\begin{equation}
    \mathbb{E}\left[\Tr'_{\tau=(13)(24)}(\bm{\xi}\bm{\xi}^{T}) \right] = \mathbb{E}\left[\Tr'_{\tau=(14)(23)}(\bm{\xi}\bm{\xi}^{T}) \right] = \mathbb{E}\left[ \Tr\left((\bm{\xi}\bm{\xi}^{T})^{2}\right) \right]
\end{equation}

Therefore by applying the formula written in Eq.\ \eqref{eq:orthogonal_weingarten_formula}, we obtain

\begin{equation} \label{eq:appli_weingarten}
\begin{alignedat}{2}
\mathbb{E}\!\left[\xi_{i_{1}}\xi_{i_{2}}\xi_{i_{3}}\xi_{i_{4}}\right] 
&= \mathbb{E}\!\left[\left(\Tr(\bm{\xi}\bm{\xi}^{T})\right)^{2}\right]
   &\;\Bigl(&\delta_{i_{1}i_{2}}\delta_{i_{3}i_{4}}Wg^{O}((1 2)(3 4),(1 2)(3 4)) \notag\\
& &&+ \delta_{i_{1}i_{3}}\delta_{i_{2}i_{4}}Wg^{O}((1 3)(2 4),(1 2)(3 4)) \notag\\
& &&+ \delta_{i_{1}i_{4}}\delta_{i_{2}i_{3}}Wg^{O}((1 4)(2 3),(1 2)(3 4))\Bigr) \notag\\[6pt]
&\quad + \mathbb{E}\!\left[\Tr\!\left((\bm{\xi}\bm{\xi}^{T})^{2}\right)\right]
   &\;\Bigl(&\delta_{i_{1}i_{2}}\delta_{i_{3}i_{4}}Wg^{O}((1 2)(3 4),(1 3)(2 4)) \notag\\
& &&+ \delta_{i_{1}i_{3}}\delta_{i_{2}i_{4}}Wg^{O}((1 3)(2 4),(1 3)(2 4)) \notag\\
& &&+ \delta_{i_{1}i_{4}}\delta_{i_{2}i_{3}}Wg^{O}((1 4)(2 3),(1 3)(2 4)) \notag\\
& &&+ \delta_{i_{1}i_{2}}\delta_{i_{3}i_{4}}Wg^{O}((1 2)(3 4),(1 4)(2 3)) \notag\\
& &&+ \delta_{i_{1}i_{3}}\delta_{i_{2}i_{4}}Wg^{O}((1 3)(2 4),(1 4)(2 3)) \notag\\
& &&+ \delta_{i_{1}i_{4}}\delta_{i_{2}i_{3}}Wg^{O}((1 4)(2 3),(1 4)(2 3))\Bigr)  
\end{alignedat}
\end{equation}

Let us also notice that actually 
\begin{equation}
\mathbb{E}\left[ \Tr\left((\bm{\xi}\bm{\xi}^{T})^{2}\right)\right]= \sum_{i,j=1}^{n} \mathbb{E}\left[ \xi_{i}^{2}\xi_{j}^{2}\right]  = \mathbb{E}\left[ \left(\Tr(\bm{\xi}\bm{\xi}^{T})\right)^{2}\right] 
\end{equation}

so that Eq.\ \eqref{eq:appli_weingarten} simplifies to
\begin{equation}
   \label{eq:noise_moments}\mathbb{E}\left[\xi_{i_{1}}\xi_{i_{2}}\xi_{i_{3}}\xi_{i_{4}}\right]= \left((\alpha+2\beta)\mathbb{E}\left[ \Tr((\bm{\xi}\bm{\xi}^{T})^{2})\right] \right) \left[\delta_{i_{1}i_{2}}\delta_{i_{3}i_{4}} + \delta_{i_{1}i_{3}}\delta_{i_{2}i_{4}} + \delta _{i_{1}i_{4}}\delta_{i_{2}i_{3}}\right]
\end{equation}
with $\alpha$ and $\beta$ the coefficients of the Weingarten matrix defined in Eq.\ \eqref{eq:weingarten_matrix}.

By substituting Eq.\ \eqref{eq:noise_moments} into Eq.\ \eqref{eq:prod_eout_before_weingarten}, we obtain

\begin{align*}
   \mathbb{E}\left[(E_\textrm{out})_{jk}(E_{\textrm{out}})_{lm}\right] &=  \frac{t_{out}-1}{t_{\textrm{out}}}\Sigma_{jk}\Sigma_{lm}\\
   &+ \frac{1}{t_{\textrm{out}}} \sum_{i_{1},i_{2},i_{3},i_{4}=1}^{n} \mathbb{E} \left( \Tilde{\Sigma}_{ji_{1}}\Tilde{\Sigma}_{i_{2}k} \Tilde{\Sigma}_{li_{3}}\Tilde{\Sigma}_{i_{4}m}\right) \left((\alpha+2\beta)\mathbb{E}\left[ \Tr((\bm{\xi}\bm{\xi}^{T})^{2})\right] \right) 
   \delta_{i_{1}i_{2}}\delta_{i_{3}i_{4}}\\
   &+ \frac{1}{t_{\textrm{out}}} \sum_{i_{1},i_{2},i_{3},i_{4}=1}^{n} \mathbb{E} \left( \Tilde{\Sigma}_{ji_{1}}\Tilde{\Sigma}_{i_{2}k} \Tilde{\Sigma}_{li_{3}}\Tilde{\Sigma}_{i_{4}m}\right) \left((\alpha+2\beta)\mathbb{E}\left[ \Tr((\bm{\xi}\bm{\xi}^{T})^{2})\right] \right) 
    \delta_{i_{1}i_{3}}\delta_{i_{2}i_{4}}\\
   &+ \frac{1}{t_{\textrm{out}}} \sum_{i_{1},i_{2},i_{3},i_{4}=1}^{n} \mathbb{E} \left( \Tilde{\Sigma}_{ji_{1}}\Tilde{\Sigma}_{i_{2}k} \Tilde{\Sigma}_{li_{3}}\Tilde{\Sigma}_{i_{4}m}\right) \left((\alpha+2\beta)\mathbb{E}\left[ \Tr((\bm{\xi}\bm{\xi}^{T})^{2})\right] \right) 
   \delta _{i_{1}i_{4}}\delta_{i_{2}i_{3}}
\end{align*}

Therefore
\begin{align}
\label{eq:expected_value_product_Eout}
\mathbb{E}\left[(E_\textrm{out})_{jk}(E_{\textrm{out}})_{lm}\right] &= \frac{t_{out}-1}{t_{\textrm{out}}}\Sigma_{jk}\Sigma_{lm} + \frac{1}{t_{\textrm{out}}}\left(\frac{1}{n(n+2)}\mathbb{E}\left[ \Tr((\bm{\xi}\bm{\xi}^{T})^{2})\right] \right)\left[ \Sigma_{jk} \Sigma_{lm} + \Sigma_{jl} \Sigma_{km} + \Sigma_{jm} \Sigma_{kl}  \right],
\end{align}
and
\begin{equation}
\label{eq:holdout_error_first_term}
    \mathbb{E}\left[\tau(\bm{\Xi}^{2})\right]= \left[ \frac{t_{\textrm{out}}-1}{t_{\textrm{out}}} + \frac{1}{t_{\textrm{out}}}\frac{3\mathbb{E}\left[ \Tr((\bm{\xi}\bm{\xi}^{T})^{2})\right]}{n(n+2)}\right]\mathbb{E}\left[\tau(\diag(\bm{V}_{\textrm{in}}^{t}\bm{\Sigma} \bm{V}_{\textrm{in}})^{2})\right].
\end{equation}

Using again similar computations, we obtain the second term appearing in the expansion of the Frobenius error written in Eq.\ \eqref{eq:frob_error_expansion}

\begin{equation}
    \label{eq:holdout_error_second_term}\mathbb{E}\left[\tau(\bm{\Xi}\bm{\Sigma})\right]=\mathbb{E}\left[\tau(\diag(\bm{V}_{\textrm{in}}^{t}\bm{\Sigma} \bm{V}_{\textrm{in}})^{2})\right].
\end{equation}
Substituting Eq.\ \eqref{eq:holdout_error_first_term} and Eq.\ \eqref{eq:holdout_error_second_term} into Eq.\ \eqref{eq:frob_error_expansion} yields the desired result.
\end{proof}

\begin{remark}
    In general, the term $\mathbb{E}\left[\tau(\diag(\bm{V}_{\textrm{in}}^{t}\bm{\Sigma} \bm{V}_{\textrm{in}})^{2})\right]$ is difficult to compute in the high-dimension limit using the Ledoit-Péché formula \cite{ledoit2011eigenvectors}. However, when the population covariance matrix $\bm{\Sigma}$ is inverse Wishart and the data are Gaussian, the computations reduce to the linear shrinkage described in Proposition \eqref{prop:linear_shrinkage}. In the following section, we will approximate the oracle eigenvalues using first a linear, then a quadratic shrinkage and will consider an inverse Wishart population covariance matrix.
\end{remark}

\subsection{Linear shrinkage approximation}

In this section, we apply our findings on the expected holdout error assuming that the population covariance matrix $\bm{\Sigma}$ belongs to the inverse Wishart ensemble to derive an approximation of the expected holdout error. Then we compute the optimal $k$ that minimizes the error of estimation under this approximation. 

When the matrix $\bm{\Sigma}$ is inverse Wishart and the data are Gaussian, the computation of $\mathbb{E}\left[\tau(\diag(\bm{V}_{\textrm{in}}^{t}\bm{\Sigma} \bm{V}_{\textrm{in}})^{2})\right]$ reduces to the optimal linear shrinkage written \eqref{eq:general_holdout_error_formula}. Although the data are not assumed Gaussian but rotationally invariant, we will approximate the oracle eigenvalues by the eigenvalues of the optimal linear shrinkage given in proposition\ \eqref{prop:linear_shrinkage}, let us define $\diag(\bm{V}_{\textrm{in}}^{t}\bm{\Sigma} \bm{V}_{\textrm{in}})_{L}$ as the diagonal matrix containing the optimal linear shrinkage eigenvalues, i.e. for $i$ an integer between $1$ and $n$
\begin{equation}
  (\diag(\bm{V}_{\textrm{in}}^{t}\bm{\Sigma} \bm{V}_{\textrm{in}})_{\textrm{L}})_{ii}=r\lambda_{i}^{\bm{E}_{\textrm{in}}}+(1-r),  
\end{equation}
where $\lambda_{i}^{\bm{E}_{\textrm{in}}}$ is the i-th largest eigenvalue of $\bm{E}_{\textrm{in}}$ and $r=\frac{p}{p+q_{\textrm{in}}}$.

Let us also define the holdout error under this linear approximation as

 \begin{equation}
 \label{eq:holdout_error_general_linear_app}
        \mathbb{E}\left(\Vert \bm{\Xi}^{H} - \bm{\Sigma} \Vert_{\textrm{L}}^{2}\right) = \left[ \frac{k}{t}\left(\frac{3 \mathbb{E}\left[ \Tr((\bm{\xi}\bm{\xi}^{T})^{2})\right]}{n(n+2)} - 1\right) - 1 \right] \mathbb{E}\left[\tau(\diag(\bm{V}_{\textrm{in}}^{t}\bm{\Sigma} \bm{V}_{\textrm{in}})_{\textrm{L}}^{2})\right]+ \mathbb{E}\left[\tau(\bm{\Sigma}^{2})\right].
    \end{equation}
    
For the purpose of computing Eq.\ \eqref{eq:holdout_error_general_linear_app} explicitly for an inverse Wishart population covariance matrix, a necessary step is the computation of the trace of the square of the sample covariance.
\begin{lemma}
\label{lemma_sample_cov_square}
Let $\bm{\xi}\in\mathbb{R}^{n}$ be a rotationally invariant noise vector such that
\begin{equation}
    \mathbb{E}\left[\bm{\xi\xi}^{T} \right]=\bm{\mathds{1}},
\end{equation}
and let $\bm{E}$ the sample covariance of $\bm{\Sigma}$ generated over $t$ data points. Let us assume that $\bm{\Sigma}$ is an inverse Wishart of parameters $n$ and $p$ such that $p$ is negligible in front of $n$, then in the limit of high dimension \begin{equation}
        \mathbb{E}\left[\tau(\bm{E}^{2}) \right]=  1+p + q\frac{\mathbb{E}\left[ \Tr((\bm{\xi}\bm{\xi}^{T})^{2})\right]}{n(n+2)}.
    \end{equation}
\end{lemma}
\begin{proof}
    \begin{equation}
        \mathbb{E}\left[ \tau(\bm{E}^{2}) \right]=\frac{1}{n}\sum_{i,j=1}^n\mathbb{E}\left[ E_{ij}^{2}\right]
    \end{equation}
Then by applying Eq.\ \eqref{eq:expected_value_product_Eout} in which we substitute $\bm{E}_{\textrm{out}}$ by $\bm{E}$ we obtain
\begin{equation}
        \mathbb{E}\left[ \tau(\bm{E}^{2}) \right]=\frac{t-1}{t}\sum_{i,j=1}^{n}\Sigma_{ij}\Sigma_{ij} + \frac{1}{t}\left(\frac{1}{n(n+2)}\mathbb{E}\left[ \Tr((\bm{\xi}\bm{\xi}^{T})^{2})\right] \right) \sum_{i,j=1}^{n}\left[ 2\Sigma_{ij} \Sigma_{ij} + \Sigma_{ii} \Sigma_{jj}  \right],
    \end{equation}

which simplifies to 
\begin{equation}
  \mathbb{E}\left[ \tau(\bm{E}^{2}) \right]  = \left(\frac{t-1}{t} + \frac{1}{t}\frac{2\mathbb{E}\left[ \Tr((\bm{\xi}\bm{\xi}^{T})^{2})\right]}{n(n+2)} \right) \mathbb{E}\left[\tau(\bm{\Sigma}^{2}) \right] + \left(\frac{n}{t}\frac{\mathbb{E}\left[ \Tr((\bm{\xi}\bm{\xi}^{T})^{2})\right]}{n(n+2)} \right) \mathbb{E}\left[\tau(\bm{\Sigma})^{2} \right]
\end{equation}

For an inverse Wishart, the following equalities hold in the high dimension limit \cite{haff1979identity,potters2020first}
\begin{equation}
    \mathbb{E}\left[\tau(\bm{\Sigma}^{2}) \right]=1+p,
\end{equation}
and
\begin{equation}
    \mathbb{E}\left[\tau(\bm{\Sigma})^{2} \right]=1.
\end{equation}

Therefore 

\begin{equation}
  \mathbb{E}\left[ \tau(\bm{E}^{2}) \right]  = \frac{t-1}{t}(1+p)  + \frac{\mathbb{E}\left[ \Tr((\bm{\xi}\bm{\xi}^{T})^{2})\right]}{n(n+2)}  \left(\frac{n}{t} + \frac{2(1+p)}{t}\right)
\end{equation}

and in the high dimension limit

\begin{equation}
  \mathbb{E}\left[ \tau(\bm{E}^{2}) \right]  = 1+p + q\frac{\mathbb{E}\left[ \Tr((\bm{\xi}\bm{\xi}^{T})^{2})\right]}{n(n+2)}.
\end{equation}
\end{proof}

Now the approximation of the holdout error written in Eq.\ \eqref{eq:holdout_error_general_linear_app} can be explicitly computed. 
\begin{corollary}
\label{cor:holdout_error_inv_wishart}
When $\Sigma$ is an inverse Wishart of parameters $n$ and $p$ such that $p$ is negligible in front of $n$ and under the assumptions of proposition\ \eqref{prop:general_holdout_error}, the expected holdout Frobenius error written in Eq.\ \eqref{eq:holdout_error_general_linear_app} in the high dimension limit equals
    \begin{equation}
    \label{eq:holdout_error_inv_wish}
        \mathbb{E}\left(\Vert \bm{\Xi}^{H} - \bm{\Sigma} \Vert_{{\textrm{L}}}^{2}\right) =\left[ \frac{k}{t}\left(\frac{3 \mathbb{E}\left[ \Tr((\bm{\xi}\bm{\xi}^{T})^{2})\right]}{n(n+2)} - 1\right) - 1 \right] \left(\frac{p^{2}}{p+\frac{kn}{kt-1}\frac{\mathbb{E}\left[ \Tr((\bm{\xi}\bm{\xi}^{T})^{2})\right]}{n(n+2)} } + 1 \right)+ 1+p.
    \end{equation}
\end{corollary}

  \begin{proof}
  Let us first compute  $\mathbb{E}\left[\tau(\diag(\bm{V}_{\textrm{in}}^{t}\bm{\Sigma} \bm{V}_{\textrm{in}})_{\textrm{L}}^{2})\right]$ in the high dimension limit 
\begin{align}
    \mathbb{E}\left[\tau(\diag(\bm{V}_{\textrm{in}}^{t}\bm{\Sigma} \bm{V}_{\textrm{in}})_{{\textrm{L}}}^{2})\right]&=\mathbb{E}\left[(r_{in}\bm{E}_{\textrm{in}}+(1-r_{\textrm{in}})\mathds{1} )^{2}\right]\notag \\ 
    &=r_{\textrm{in}}^{2}\mathbb{E}\left[\tau(\bm{E}_{\textrm{in}}^{2})\right] + r_{\textrm{in}}(1-r_{\textrm{in}})\mathbb{E}\left[\tau(\bm{E}_{\textrm{in}}) \right] + (1-r_{\textrm{in}})^{2} \notag \\
    &= r_{\textrm{in}}^{2}(\mathbb{E}\left[\tau(\bm{E}_{\textrm{in}}^{2})\right]-1) + 1.
\end{align} 

with
\begin{equation}
    r_{\textrm{in}}= \frac{p}{p+\frac{kn}{kt-1}\frac{\mathbb{E}\left[ \Tr((\bm{\xi}\bm{\xi}^{T})^{2})\right]}{n(n+2)} }
\end{equation}

Then by using lemma\ \eqref{lemma_sample_cov_square} which derives the second order moment of the sample covariance of an inverse Wishart, we obtain

\begin{align}
    \label{eq:oracle_inv_wish}\mathbb{E}\left[\tau(\diag(\bm{V}_{\textrm{in}}^{t}\bm{\Sigma} \bm{V}_{\textrm{in}})_{{\textrm{L}}}^{2})\right]&=\frac{p^{2}}{p+\frac{kn}{kt-1}\frac{\mathbb{E}\left[ \Tr((\bm{\xi}\bm{\xi}^{T})^{2})\right]}{n(n+2)} } + 1.
\end{align}

By injecting Eq.\ \eqref{eq:oracle_inv_wish} to  Eq.\ \eqref{eq:general_holdout_error_formula}, we obtain the desired result, Eq. \eqref{eq:holdout_error_inv_wish}.
  \end{proof}

\begin{remark}
    Let us now comment the range of the ratio 

\begin{equation}
    \gamma:=\frac{\mathbb{E}\left[ \Tr((\bm{\xi}\bm{\xi}^{T})^{2})\right]}{n(n+2)}
\end{equation}
appearing in the holdout error formula given in Eq.\ \eqref{eq:general_holdout_error_formula}. Because of the assumption written in proposition\ \eqref{prop:general_holdout_error}
    \begin{equation}
        \mathbb{E}\left[ \bm{\xi}\bm{\xi}^{T}\right]=\bm{\mathds{1}}
    \end{equation}
    and by using Jensen inequality, we obtain that 
\begin{equation}
    \gamma\geq \frac{n^{2}}{n(n+2)}.
\end{equation}
Therefore, $\gamma$ is at least in $\mathcal{O}(1)$ in the high dimension limit. 
\end{remark}
\subsubsection{Optimal split}
\begin{corollary}

Under the assumptions of corollary\ \eqref{cor:holdout_error_inv_wishart}, there exists a $k$, the train-test ratio, defined from $t_{\textrm{out}}$ an integer between $1$ and $t-1$, that minimizes the expected approximated Frobenius holdout error of estimation written in Eq.\ \eqref{eq:holdout_error_inv_wish}, equal to
    \begin{align*}
        k_{\textrm{opt}}&= 
\frac{p \left( \gamma^2 n q^2 \cdot (3\gamma - 1) + n p^2 \cdot (3\gamma p + 3\gamma - p - 1)\right)}{n (\gamma q + p) \left( \gamma^2 q^2 \cdot (3\gamma - 1) + \gamma p q (3\gamma p + 6\gamma - p - 2) + p^2 \cdot (3\gamma p + 3\gamma - p - 1) \right)}\\
&+ \frac{p\left(q \left( 3\gamma^2 n p^2 + 6\gamma^2 n p - \gamma n p^2 - 2\gamma n p + \sqrt{\frac{\gamma n^2 \cdot (3\gamma - 1)(n p + q(\gamma n - 3\gamma p + p))(\gamma q + p^2 + p)}{q^2}} (\gamma q + p) \right) \right)}
{n (\gamma q + p) \left( \gamma^2 q^2 \cdot (3\gamma - 1) + \gamma p q (3\gamma p + 6\gamma - p - 2) + p^2 \cdot (3\gamma p + 3\gamma - p - 1) \right)},
    \end{align*}
\end{corollary}
\begin{proof}
    This equality is obtained directly by canceling the derivative of the expected Frobenius error in Eq. \eqref{eq:holdout_error_inv_wish}. 
\end{proof}

\begin{remark}
    $k_{\textrm{opt}}$ exists only if $\gamma>\frac{1}{3}$, however, this condition is automatically verified because of the assumption 
    \begin{equation}
        \mathbb{E}\left[\bm{\xi}\bm{\xi}^{T} \right]=\bm{\mathds{1}}.
    \end{equation}

\end{remark}

When $n\to\infty$, one finds the following
\begin{equation}
\label{eq:kopt}
    k_{\textrm{opt}}\sim  \frac{\sqrt{\gamma (3\gamma - 1)(p+q\gamma(\gamma q +p^{2}+p))}}{\gamma^{2}q^{2}(3\gamma -1) + \gamma pq(3\gamma p + 6\gamma - p - 2)+p^{2} (3\gamma p + 3\gamma - p -1) }   p\sqrt{n}.
\end{equation}
Therefore, the optimal $k$, the train-test ratio, which minimizes the expected holdout error, is proportional to the square root of the dimension of the population covariance matrix.

\begin{figure}
\centering
        \includegraphics[width=0.6\columnwidth]{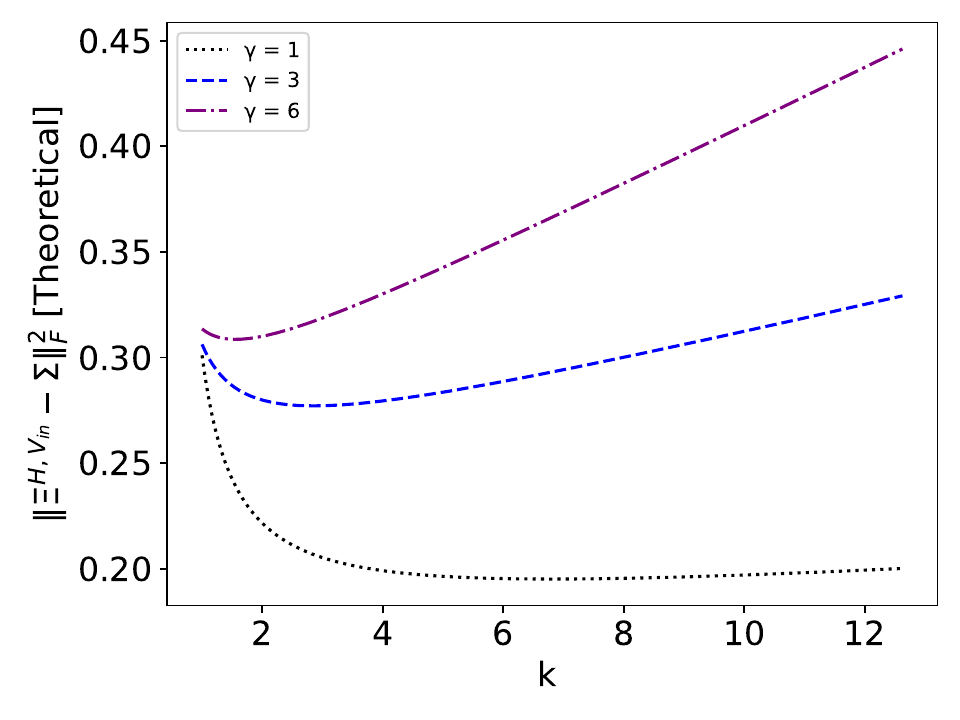}  
    \caption{Holdout theoretical error from equation~\eqref{eq:holdout_error_inv_wish} as a function of $k=\frac{t}{t_{\textrm{out}}}$ where $1\leq t_{\textrm{out}}\leq t-1$ for an inverse Wishart population matrix with parameters  $n=500$, $p=0.3$, $q=0.4$ where $\gamma=1$ corresponds to multivariate Gaussian noise vectors, $\gamma=3$ corresponds to the norm of the noise vector following a centered Gaussian of variance $\sqrt{n}$ and $\gamma=6$ corresponds to the norm of the noise following a centered Laplace distribution of scale parameter $\sqrt{\frac{n}{2}}$. } 
    \label{fig:theoretical_holdout_error_three_gammas}
\end{figure} 

Let us also consider the evolution of the holdout error as a function of $\gamma$: Figure \ref{fig:theoretical_holdout_error_three_gammas} plots the expected holdout error from Eq. \eqref{eq:holdout_error_inv_wish} as a function of $k$ for an inverse Wishart population matrix, fixed $n$, $q$ and $p$ and for $\gamma=\{1,3,6\}$. The values of $\gamma$ correspond to multivariate Gaussian data, a centered Gaussian scaling with a $\sqrt{n}$ variance, and a Laplace scaling distributions, respectively (see Section \ref{Section 4}). Note that the expected holdout error around $k_{\textrm{opt}}$ is shallow for $\gamma=1$ and becomes sharper as $\gamma$ increases. In practice, this implies that the choice of $k$ is more important for large values of $\gamma$.

\subsection{Quadratic shrinkage approximation}

In this section, we consider a polynomial of order two in the sample eigenvalues to approximate the oracle eigenvalues. In a similar manner as in the previous section, we will consider the optimal quadratic shrinkage that minimizes the Frobenius norm with respect to the population covariance $\bm{\Sigma}$ which amounts to \cite{ledoit2022quadratic}
    \begin{equation*}
\label{eq:quadratic_shrinkage_formulation}
    \underset{\alpha_{1},\alpha_{2},\alpha_{3}\in[0,1]}{\mathrm{argmin}}\Vert \alpha_{1}\bm{E}^{2}+\alpha_{2}\bm{E}+\alpha_{3}\mathds{1} - \bm{\Sigma} \Vert^{2} 
\end{equation*}
under the constraint
\begin{equation}
 \tau(\alpha_{1}\bm{E}^{2}+\alpha_{2}\bm{E}+\alpha_{3}\mathds{1})=1,
\end{equation}

The solution to this problem is given in the following proposition.

\begin{proposition}
Let us consider $\bm{\Sigma}$ a positive semi-definite matrix and $\bm{E}$ the sample covariance estimator of $\bm{\Sigma}$. We assume that both $\bm{\Sigma}$ and $\bm{E}$ are normalized such that $\mathbb{E}\left[\tau(\bm{\Sigma})\right]=\mathbb{E}\left[\tau(\bm{E})\right]$=1 and that $\bm{E}$ has finite tracial moments up to order $4$. Then the optimal quadratic shrinkage coefficients that minimize the expected Frobenius error of estimation given in Eq.\ \eqref{eq:quadratic_shrinkage_formulation} are
\cite{ledoit2022quadratic}
\begin{equation}
\label{eq:quad_shrinkage_1st_coeff}
    \alpha_{1}=\frac{\left( \mathbb{E}\left[\tau(\bm{E}^{2}) \right]- \mathbb{E}\left[ \tau(\bm{E}^{3}) \right]\right) \left( \mathbb{E}\left[\tau(\bm{E}\bm{\Sigma})\right]-1 \right) + \left(\mathbb{E}\left[\tau(\bm{E}^{2}\bm{\Sigma})\right] - \mathbb{E}\left[ \tau(\bm{E}^{2})\right] \right) \left(\mathbb{E}\left[ \tau(\bm{E}^{2})\right] - \mathbb{E}\left[ \tau(\bm{E})\right] \right)}{\left( \mathbb{E}\left[\tau(\bm{E}^{2}) \right]- \mathbb{E}\left[ \tau(\bm{E}) \right]\right) \left( \mathbb{E}\left[\tau(\bm{E}^{4})\right]-\mathbb{E}\left[\tau(\bm{E}^{2})^{2} \right] \right) - \left( \mathbb{E}\left[\tau(\bm{E}^{3}) \right]-\mathbb{E}\left[\tau(\bm{E}^{2})\right] \right)^{2}}
\end{equation}

\begin{equation}
 \label{eq:quad_shrinkage_2nd_coeff}   \alpha_{2}=\frac{\mathbb{E}\left[\tau(\bm{E}\bm{\Sigma}) \right]-1 -\alpha_{1}\left(\mathbb{E}\left[\tau(\bm{E}^{3})\right]- \mathbb{E}\left[\tau(\bm{E}^{2})\right]\right)}{\mathbb{E}\left[\tau(\bm{E}^{2})\right]- \mathbb{E}\left[\tau(\bm{E})\right]}
\end{equation}

\begin{equation}
\label{eq:quad_shrinkage_3rd_coeff}
    \alpha_{3}=1-\alpha_{1}\mathbb{E}\left[ \tau(\bm{E}^{2})\right]-\alpha_{2}\mathbb{E}\left[ \tau(\bm{E}) \right].
\end{equation}

\end{proposition}

\begin{proof}

Let us note $\lambda$ the Lagrangian associated with the constraint $\tau(\alpha_{1}\bm{E}^{2}+\alpha_{2}\bm{E}+\alpha_{3}\mathds{1})=1$ and $E(\alpha_{1},\alpha_{2},\alpha_{3},\lambda)$ the following functional
\begin{align}
 E(\alpha_{1},\alpha_{2},\alpha_{3},\lambda)&=  \alpha_{1}^{2}\tau(\bm{E}^{4})+  \alpha_{2}^{2}\tau(\bm{E}^{3}) + \alpha_{3}^{2}\tau(\mathds{1})+\tau(\bm{\Sigma}^{2})+2\alpha_{1}\alpha_{2}\tau(\bm{E}^{3})+2\alpha_{1}\alpha_{3}\tau(\bm{E}^{2})\\ \notag
 &-2\alpha_{1}\tau(\bm{E}^{2}\bm{\Sigma})+2\alpha_{2}\alpha_{3}\tau(\bm{E})-2\alpha_{2}\tau(\bm{E\Sigma})-2\alpha_{3}\tau(\bm{\Sigma})\\ \notag
 &+\lambda(1-\alpha_{1}\tau(\bm{E}^{2})-\alpha_{2}\tau(\bm{E})-\alpha_{3}\tau(\mathds{1}))
\end{align}
The optimality conditions are the three following equations
    \begin{equation}
        \frac{\partial E}{\partial\alpha_{1}}=0 \Leftrightarrow 2\alpha_{1}\tau(\bm{E}^{4})+2\alpha_{2}\tau(\bm{E}^{3})+2\alpha_{3}\tau(\bm{E}^{2})-2\tau(\bm{E}^{2}\bm{\Sigma})-\lambda\tau(\bm{E}^{2})=0,
    \end{equation}

    \begin{equation}
        \frac{\partial E}{\partial\alpha_{2}}=0 \Leftrightarrow 2\alpha_{2}\tau(\bm{E}^{2})+2\alpha_{1}\tau(\bm{E}^{3})+2\alpha_{3}-2\tau(\bm{E\Sigma})-\lambda=0
    \end{equation}
and 
    \begin{equation}
        \frac{\partial E}{\partial\alpha_{3}}=0 \Leftrightarrow 2\alpha_{3}+2\alpha_{1}\tau(\bm{E}^{2})+2\alpha_{2}-2-\lambda=0,
    \end{equation}

and a few successive substitutions give the announced result.
\end{proof}

Now, let us define $\diag(\bm{V}_{\textrm{in}}^{t}\bm{\Sigma} \bm{V}_{\textrm{in}})_{\textrm{Q}}$ the diagonal matrix containing the optimal quadratic shrinkage eigenvalues, i.e. for $i$ an integer between $1$ and $n$

\begin{equation}
    (\diag(\bm{V}_{\textrm{in}}\bm{\Sigma}\bm{V}_{\textrm{in}}^{T}))_{\textrm{Q},ii}=\alpha_{1}(\lambda_{i}^{\bm{E}_{\textrm{in}}})^{2}+\alpha_{2}\lambda_{i}^{\bm{E}_{\textrm{in}}}+\alpha_{3},
\end{equation}
where $\lambda_{i}^{\bm{E}_{\textrm{in}}}$ is the i-th largest eigenvalue of $\bm{E}_{\textrm{in}}$ and $\alpha_{1}$, $\alpha_{2}$ and $\alpha_{3}$ are the coefficients written in Eq.\ \eqref{eq:quad_shrinkage_1st_coeff}, Eq.\ \eqref{eq:quad_shrinkage_2nd_coeff} and Eq.\ \eqref{eq:quad_shrinkage_3rd_coeff} respectively.

Let us also define the holdout error under the quadratic approximation as

 \begin{equation}
 \label{eq:holdout_error_general_quad_app}
        \mathbb{E}\left(\Vert \bm{\Xi}^{H} - \bm{\Sigma} \Vert_{\textrm{Q}}^{2}\right) = \left[ \frac{k}{t}\left(\frac{3 \mathbb{E}\left[ \Tr((\bm{\xi}\bm{\xi}^{T})^{2})\right]}{n(n+2)} - 1\right) - 1 \right] \mathbb{E}\left[\tau(\diag(\bm{V}_{\textrm{in}}^{t}\bm{\Sigma} \bm{V}_{\textrm{in}})_{\textrm{Q}}^{2})\right]+ \mathbb{E}\left[\tau(\bm{\Sigma}^{2})\right].
    \end{equation}
    
\begin{remark}
For the detailed computations of the third and fourth order tracial moments of $\bm{E}$, i.e. $\mathbb{E}\left[\tau(\bm{E}^{3}) \right]$ and $\mathbb{E}\left[\tau(\bm{E}^{4}) \right]$, that are needed to compute the optimal quadratic shrinkage coefficients and their error see the Appendix \ref{section:appendix}.
\end{remark}

\begin{corollary}
\label{cor:holdout_error_quad_inv_wishart}
When $\Sigma$ is an inverse Wishart of parameters $n$ and $p$ such that $p$ is negligible in front of $n$, under the assumptions of proposition\ \eqref{prop:general_holdout_error} and by assuming that 
\begin{equation}
    \mathbb{E}\left[\Vert\bm{\xi}\Vert^{6} \right]<\infty
\end{equation}
and 
\begin{equation}
    \mathbb{E}\left[\Vert\bm{\xi}\Vert^{8} \right]<\infty,
\end{equation}

 the expected holdout Frobenius error written in Eq.\ \eqref{eq:holdout_error_general_quad_app} in the high dimension limit equals
    \begin{equation}
    \label{eq:holdout_error_quad_inv_wish}
        \mathbb{E}\left(\Vert \bm{\Xi}^{H} - \bm{\Sigma} \Vert_{\textrm{Q}}^{2}\right) =\left[ \frac{k}{t}\left(\frac{3 \mathbb{E}\left[ \Tr((\bm{\xi}\bm{\xi}^{T})^{2})\right]}{n(n+2)} - 1\right) - 1 \right] \mathcal{E} +  1+p,
    \end{equation}

with
\begin{equation}
    \mathcal{E}=\alpha_{1}^{2}\left[\mathbb{E}\left[\tau(\bm{E}^{4})\right] - \mathbb{E}\left[\tau(\bm{E}^{2})^{2}\right] - r_{2}^{2}\left(\mathbb{E}\left[\tau(\bm{E}^{2})\right] - 1\right) \right] + r^{2}\left[\mathbb{E}\left[\tau(\bm{E}^{2})\right] -1\right] +1 ,
\end{equation}

\begin{equation}
    r=\frac{p}{p+q\frac{\mathbb{E}\left[\Tr((\bm{\xi}\bm{\xi}^{T})^{2}) \right]}{n(n+2)}}   \hspace{5mm} \hbox{and} \hspace{5mm} r_{2}=\frac{\mathbb{E}\left[\tau(\bm{E}^{3}) \right] - \mathbb{E}\left[\tau(\bm{E}^{2}) \right]}{p+q\frac{\mathbb{E}\left[\Tr((\bm{\xi}\bm{\xi}^{T})^{2}) \right]}{n(n+2)}}.
\end{equation}
\end{corollary}

\section{Applications to rotationally invariant noise distributions}
\label{Section 4}

In this section, we compute the expected holdout error for different rotationally invariant noise distributions. We begin by considering uniform and Gaussian noise which have $\gamma=1$. Then, we consider distributions that lead to different norms of noise, namely Gaussian, Student and Laplace to illustrate the influence of the factor $\gamma$.

Let us first mention that any rotationally invariant vector $\bm{\xi}$ can be written as the following product

\begin{equation}
    \bm{\xi}=s\bm{u},
\end{equation}
where $\bm{u}\in\mathbb{R}^{n}$ is a uniform vector on the sphere $\mathcal{S}^{n-1}$ and $s$ is a random variable independent from $\bm{u}$.
\begin{equation}
    \mathbb{E}\left[\Tr(\bm{\xi\xi}^{T})^{2} \right]=\mathbb{E}[ s^{4}] \mathbb{E}\left[\Tr(\bm{uu}^{T})^{2} \right],
\end{equation}

with $\bm{u}\in\mathbb{R}^{n}$ a uniform vector in the sphere $\mathcal{S}^{n-1}$ independent of $x$ a real random variable.

\begin{equation}
\label{eq:uniform_distrib_alpha}
    \mathbb{E}\left[\Tr((\bm{u}\bm{u}^{T})^{2} \right]=n(n-1)\mathbb{E}\left[u_{1}^{2}u_{2}^{2} \right] + n \mathbb{E}\left[u_{1}^{4}\right].
\end{equation}

We have
\begin{equation}
  \mathbb{E}\left[u_{1}^{2}u_{2}^{2} \right] = \frac{n^{2}}{n(n-1)}  
\end{equation}
and
\begin{equation}
  \mathbb{E}\left[u_{1}^{4} \right] = \frac{3}{n(n-1)}.  
\end{equation}

Then by reinjecting these moments into Eq. \eqref{eq:uniform_distrib_alpha}, we obtain the following.

\begin{equation}
\label{eq:unif_alpha_1}
    \mathbb{E}\left[\Tr((\bm{u}\bm{u}^{T})^{2} \right]=1+\frac{3}{n+1}.
\end{equation}

To apply our findings on the holdout error written in Eq.\ \eqref{eq:holdout_error_inv_wish}, we therefore need
\begin{equation}
\label{eq:condition_on_s}
    \mathbb{E}\left[s^{2} \right]=n.
\end{equation}

\subsection{Gaussian multivariate noise}

The multivariate Gaussian distribution is rotationally invariant, so one can consider the following noise vector: $\bm{\xi}\sim \mathcal{N}(0,\bm{\mathds{1}})$ with $\bm{\mathds{1}}$ being the identity matrix of size $n$. This case corresponds to $s$ following the square root of a $\chi^{2}$ distribution with $n$ degrees of freedom.

In this setting
\begin{align*}
    \mathbb{E}\left[\Tr((\bm{\xi}\bm{\xi}^{T})^{2})\right]&=\mathbb{E}\left[\sum \xi_{i}^{2}\xi_{j}^{2}\right] \\
    & = n(n-1) + 3n\\
    & = n(n+2).
\end{align*}

Therefore Eq.\ \eqref{eq:general_holdout_error_formula} becomes
    \begin{equation}
        \mathbb{E}\left[\Vert \bm{\Xi}^{H} - \bm{\Sigma} \Vert_{F}^{2} \right] = \left[ \frac{2}{t_{\textrm{out}}} - 1 \right] \mathbb{E}\left[\tau(\diag(\bm{V}_{\textrm{in}}^{T}\bm{\Sigma} \bm{V}_{\textrm{in}})^{2}))\right] +  \mathbb{E}\left[\tau(\bm{\Sigma}^{2})\right],
    \end{equation}

which corresponds to the result found in the Gaussian case directly using the Wick theorem; see \cite{lam2024holdout} for more details.

\subsection{Uniform multivariate noise}

Let us consider a constant $\bm{\xi}$ a uniform vector on the sphere of radius $\sqrt{n}$, then
\begin{equation}
     \lim\limits_{n \to \infty}
   \frac{\mathbb{E}\left[ \Tr((\bm{\xi}\bm{\xi}^{T})^{2})\right]}{n(n+2)}=1,
\end{equation}

and we obtain the same formulae as in the Gaussian multiplicative noise case.

\subsection{Gaussian, Student and Laplace noise norms}

Let us now consider several examples for the distribution of the random variable $s$ which represents the norm of the noise.

\subsubsection*{Gaussian norm}
Let us assume $s\sim \mathcal{N}(0,\sqrt{n)}$, then
\begin{equation}
     \lim\limits_{n \to \infty}
   \frac{\mathbb{E}\left[ \Tr((\bm{\xi}\bm{\xi}^{T})^{2})\right]}{n(n+2)}=3.
\end{equation}

\subsubsection*{Student scaling}
Let us assume that $s=\sqrt{\frac{\nu-2}{\nu}n}\delta$ where $\delta$ follows a Student-t distribution with $\nu >4$ the degree of freedom, then
\begin{equation}
    \lim\limits_{n \to \infty}
   \frac{\mathbb{E}\left[ \Tr((\bm{\xi}\bm{\xi}^{T})^{2})\right]}{n(n+2)}=\frac{3(\nu-2)}{\nu-4}.
\end{equation}

\subsubsection*{Laplace norm}

Let us assume that $s$ follows a centered Laplace distribution with a scale of $\sqrt{\frac{n}{2}}$, then

\begin{equation}
    \lim\limits_{n \to \infty}
   \frac{\mathbb{E}\left[ \Tr((\bm{\xi}\bm{\xi}^{T})^{2})\right]}{n(n+2)}=6,
\end{equation}
which also corresponds to a Student-t distribution with $\nu=6$, but in this case the 6-th and 8-th moments of the distribution would diverge and the quadratic shrinkage cannot be applied.

\begin{figure} 
    \centering
    \subfigure{
        \includegraphics[width=0.3\textwidth]{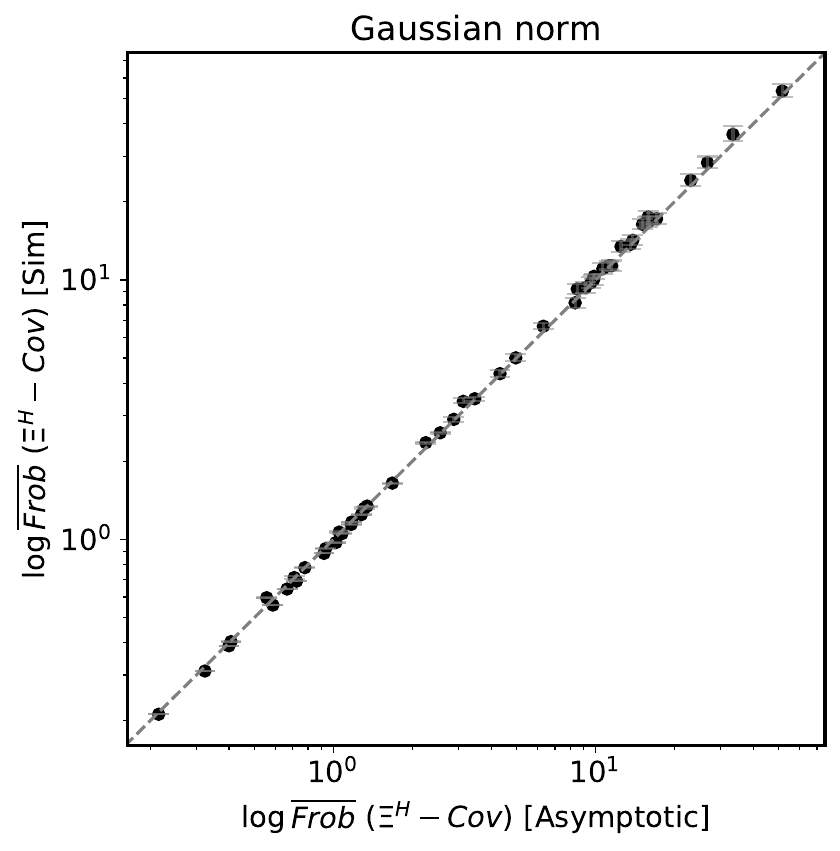}
    }
    \subfigure{
\includegraphics[width=0.3\textwidth]{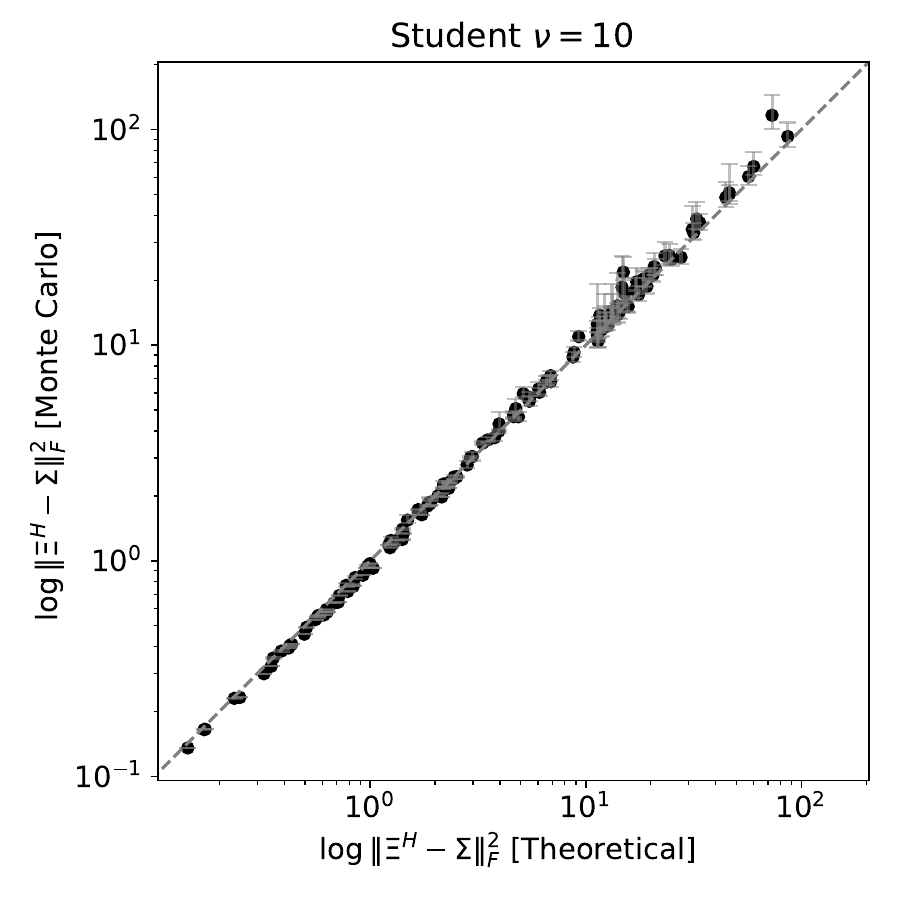}
    }
    
    \subfigure{
        \includegraphics[width=0.3\textwidth]{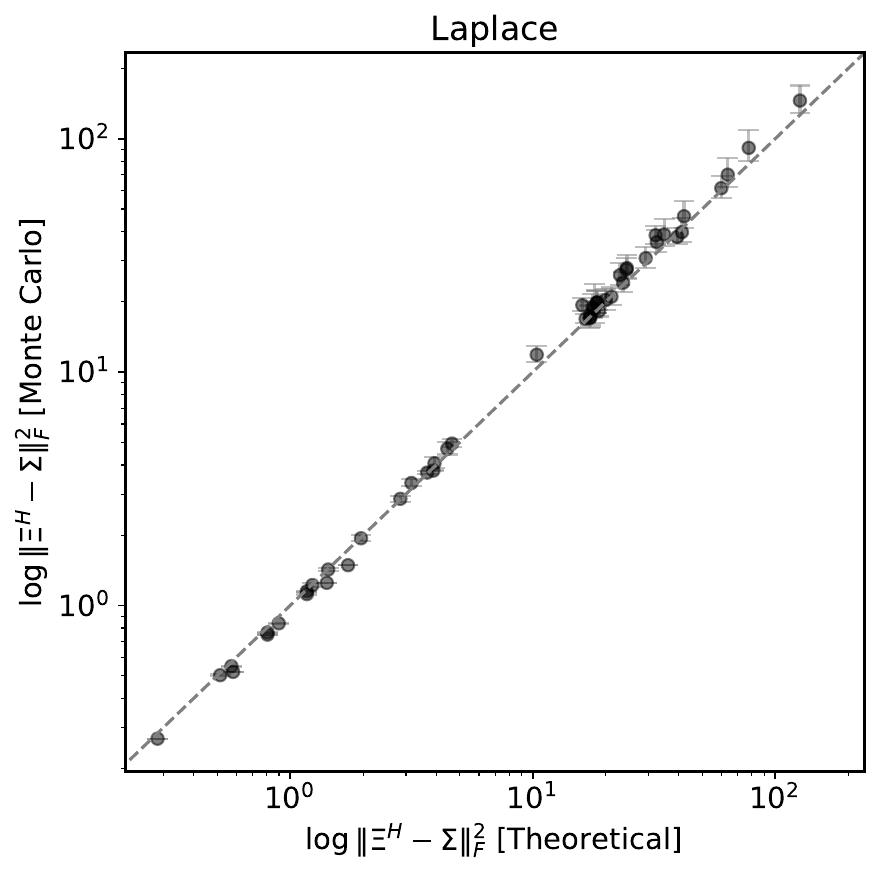}
    }
    \subfigure{
\includegraphics[width=0.3\textwidth]{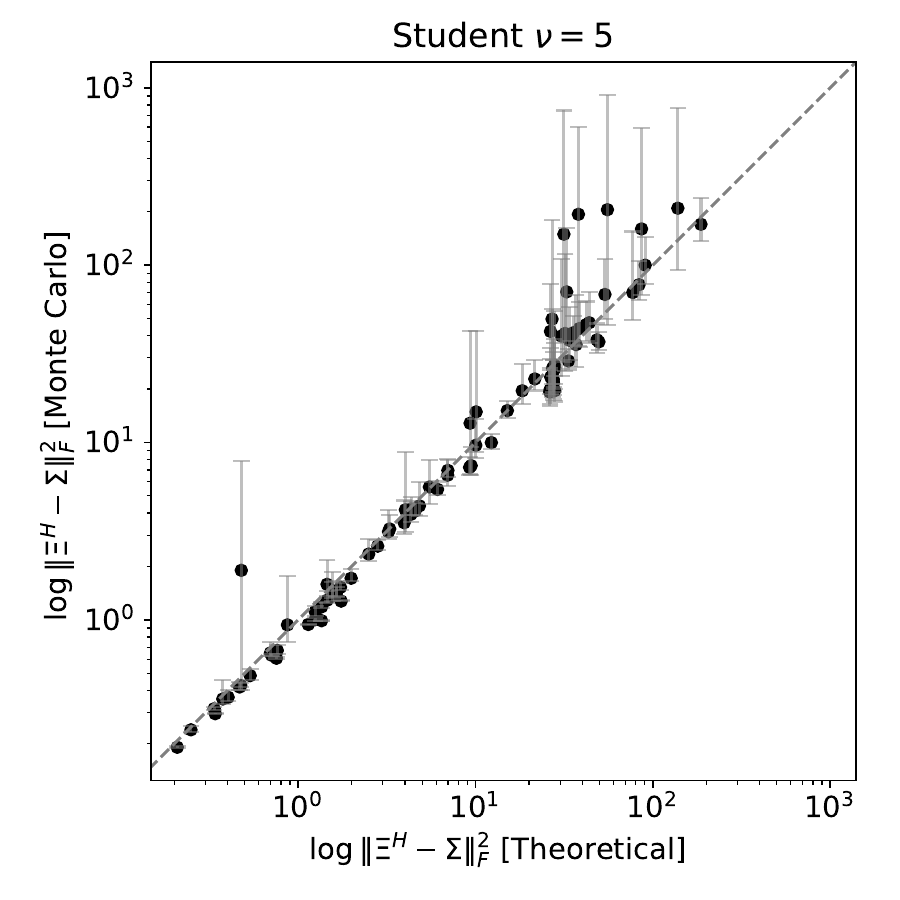}
    }
    \caption{The plots compare the the theoretical error of equation~\eqref{eq:holdout_error_inv_wish} with a Monte Carlo estimations for a random uniform selection of the parameters: $n \in [ 100,1000 ]$, $p\in [0.1,9 ]$, $q\in [0.1,0.9]$ and $k$ among the divisors of $t$. For stability, the Monte Carlo is averaged over $20000$ simulations for different distributions of the norm of the noise, namely Gaussian, Student and Laplace as defined in \ref{Section 4}. The confidence intervals displayed are obtained by a BCa bootstrap with a $95\%$ confidence level.}
    \label{fig:holdout_monte_carlo_vs_theoretical_gaussian_laplace_student_streched_exp}
\end{figure}

All distributions shown in Figure \ref{fig:holdout_monte_carlo_vs_theoretical_gaussian_laplace_student_streched_exp} have a $\gamma$ larger than $1$. In this case, the dispersion of the error around its expectation is very large, especially for higher values of $k$. To verify the pertinence of our approach numerically, we performed Monte Carlo simulations by choosing the values of $n$ uniformly in $\left[100,1000\right]$, $q\in\left[0.1,0.9\right]$ and $k$, the train-test ratio, among the divisors of $t$. For each set of parameters, the Monte Carlo holdout error is the average over $20000$ independent realizations. Then we added $95\%$ confidence intervals obtained by BCa bootstrap and compared the Monte Carlo holdout error to the theoretical holdout error of equation \eqref{eq:holdout_error_inv_wish}. 

We observe a small negative bias in the simulations which becomes more visible for higher values of $\gamma$. Indeed, the theoretical error we derived approximating the oracle eigenvalues by the linear shrinkage appears larger than the average realized holdout error. Furthermore, the higher the value of $\gamma$, the higher the dispersion of the error around its mean making the choice of the optimal $k$ more important.

\begin{figure} 
    \centering
    \subfigure{
        \includegraphics[width=0.3\textwidth]{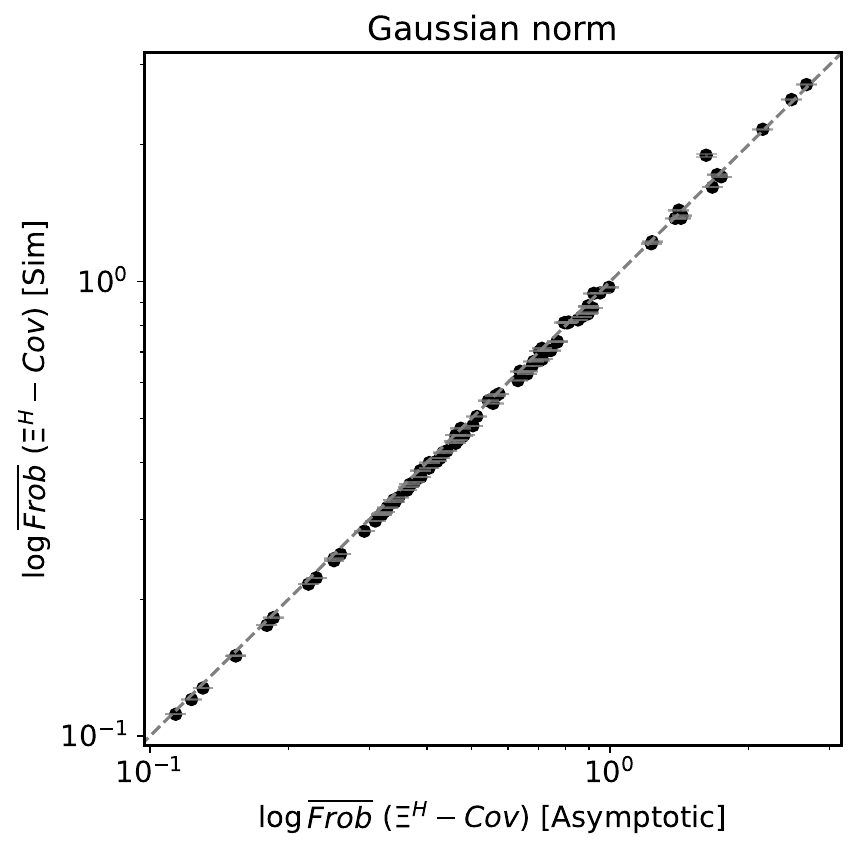}
    }
    \subfigure{
\includegraphics[width=0.3\textwidth]{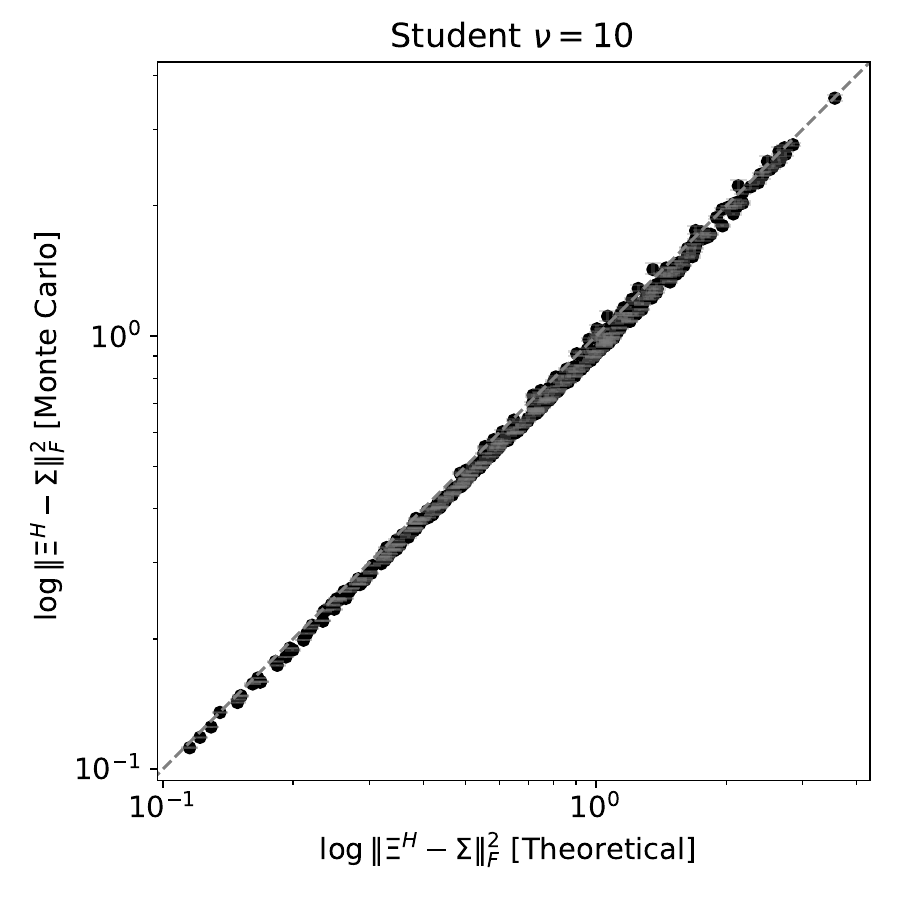}
    }

    \subfigure{
        \includegraphics[width=0.3\textwidth]{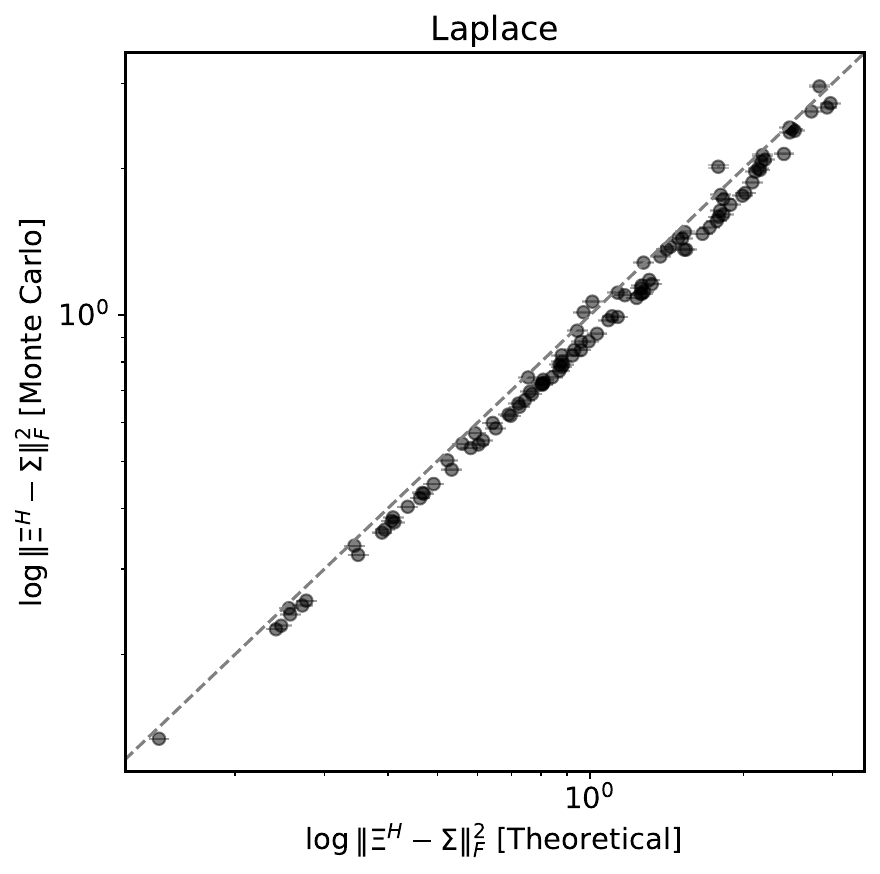}
    }
    \subfigure{
    \includegraphics[width=0.3\textwidth]{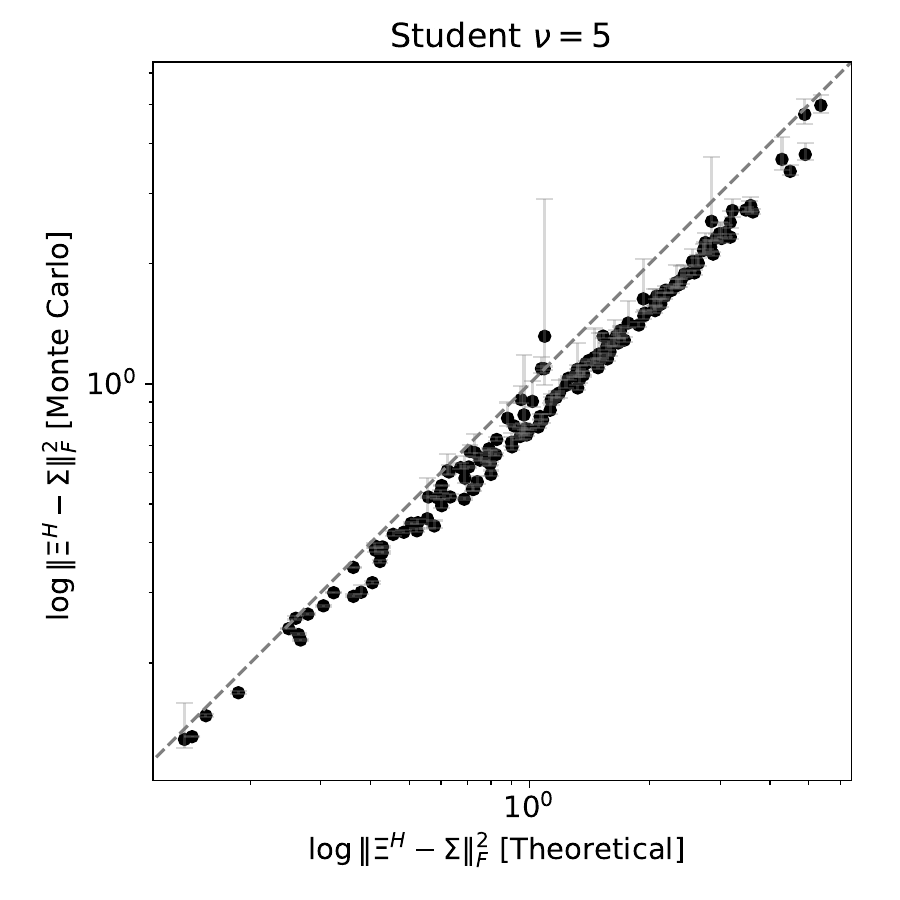}
    }
    
    \caption{The plots compare the theoretical error of equation~\eqref{eq:holdout_error_inv_wish} with a Monte Carlo estimates for a random uniform selection of the parameters: $n \in [ 100,1000 ]$, $p\in [0.1,9 ]$, $q\in [0.1,0.9]$ and $k$ chosen as the optimal value of Eq.\ \eqref{eq:kopt}. For stability, the Monte Carlo estimates is averaged over $20000$ simulations for different distributions of the norm of the noise, namely Gaussian, Student and Laplace as defined in \ref{Section 4}. The confidence intervals displayed are obtained by a BCa bootstrap with a $95\%$ confidence level.}
    \label{fig:holdout_monte_carlo_vs_theoretical_kopt}
\end{figure}

Figure \ref{fig:holdout_monte_carlo_vs_theoretical_kopt} compares the theoretical holdout error obtained with the linear approximation written in Eq.\ \eqref{eq:holdout_error_inv_wish} with a Monte Carlo holdout error for $k$ chosen as in Eq.\ \eqref{eq:kopt}, i.e. the optimal $k$ for the linear shrinkage approximation of the holdout error. We observe that the negative bias of estimation is higher for smaller values of $k$ as well as for higher gammas.

\begin{figure} 
    \centering
    \subfigure{
        \includegraphics[width=0.3\textwidth]{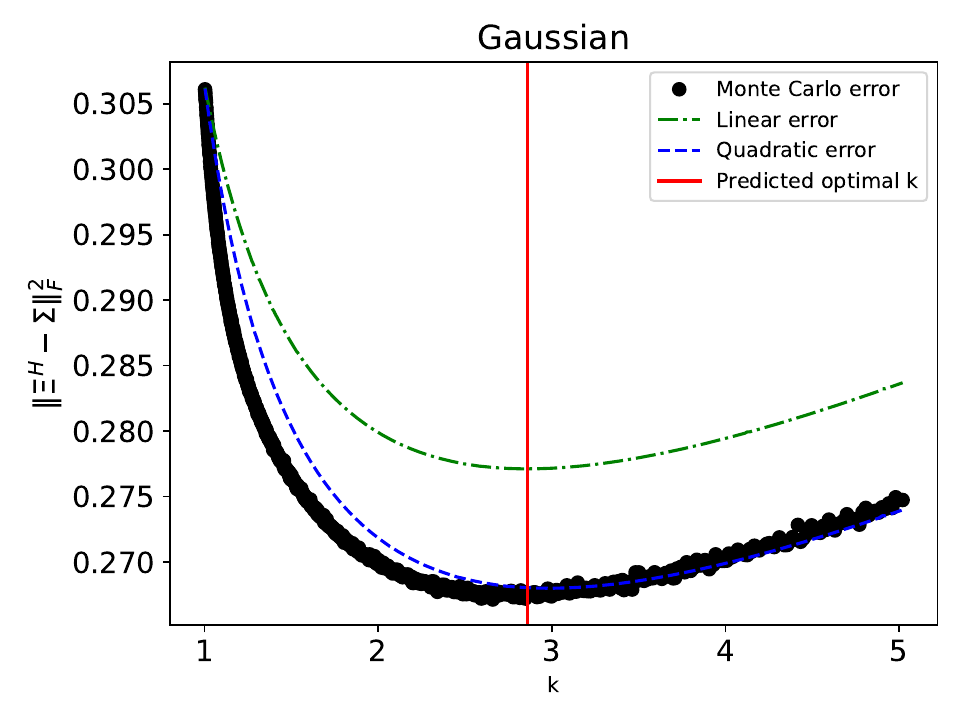}
    }
    \subfigure{
\includegraphics[width=0.3\textwidth]{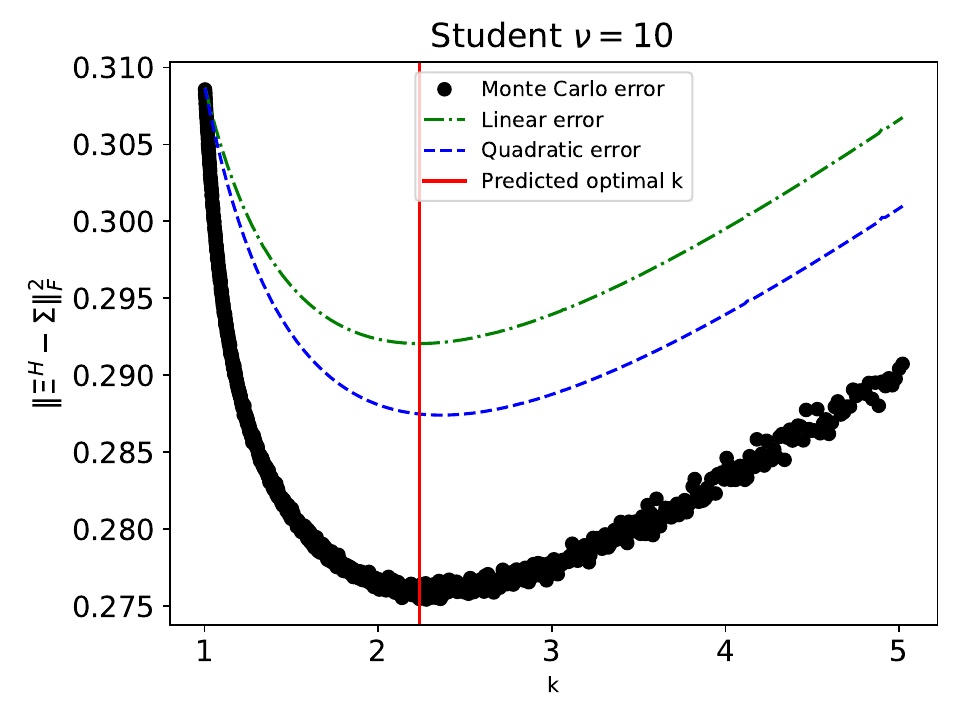}
    }
    
    \subfigure{
        \includegraphics[width=0.3\textwidth]{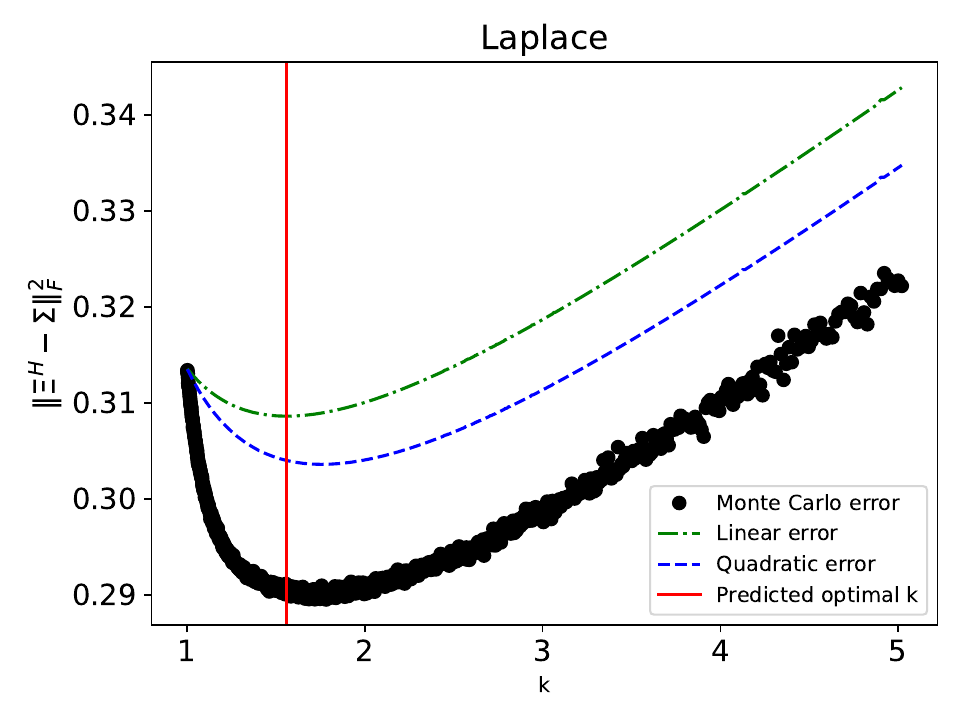}
    }
   
    \caption{The plots compare the holdout errors obtained by approximating the oracle eigenvalues with a linear and a quadratic polynomial of the sample eigenvalues to a Monte Carlo holdout error. The population covariance matrix is inverse Wishart with $n=500$ and $p=0.3$. The data are generated with a rotationally invariant noise where the norm follows a Gaussian, Student and Laplace distribution with $q=0.4$.}
    \label{fig:holdout_all_k_comparison_lin_quad_mc}
\end{figure}

In Figure \ref{fig:holdout_all_k_comparison_lin_quad_mc}, the plots compare the linear and quadratic shrinkage holdout errors to the realized holdout error. It seems that the quadratic shrinkage corrects a significant part of this negative bias, especially for values of $k$ around the realized optimal one and seems to be a particularly good approximation when the norm of the noise vector follows a gaussian distribution. Moreover, the optimal $k$ derived with the linear approximation written in Eq.\ \eqref{eq:kopt} seems to approximate well the realized optimal $k$, further strengthening the linear approximation. 

Let us also mention that the Student distribution with parameter $\nu=5$ could not be considered for performing the quadratic shrinkage as it does not have finite six and eight order moments. The lack of finite higher moments for certain distributions may lead one to consider other functions of the sample eigenvalues to approximate the oracle eigenvalues instead of using polynomial functions. One possibility would be to use rational functions.

\section{Conclusion}
In this work, we derive the expected Frobenius holdout error for rotationally invariant multiplicative noise models. Hereby extending to non-Gaussian distributions for the noise vector. Our approach involved using Weingarten calculus to link the moments of sample covariances to the moments of the population covariance matrix. This step enables us to write the expected Frobenius holdout error as a function of the oracle eigenvalues which can then be explicitly derived, in the high-dimension limit, using the Ledoit-Péché's formula when its assumptions are satisfied. 

For general population covariance matrices, approximations are necessary, for example, to estimate the limiting spectral distribution of the sample covariance, and such work is left for further research. When the population covariance matrix follows an inverse Wishart distribution, we approximate the oracle eigenvalues using a linear shrinkage considering that such an approximation becomes exact for Gaussian data and we obtain the expression of the optimal $k$, the ratio between the number of data observations and the size of the out-of-sample set, finding it to be proportional to the square root of the matrix dimension. Then, we approximate the oracle eigenvalues using a quadratic shrinkage for which numerical methods could be of use to approximate an optimal $k$. The quadratic approximation yields a better performance than the linear approximation, especially for values of $k$, the optimal train-test split ratio, around the optimal realized one.

Finally, we were able to test different noise distributions for the norm of the noise vector, namely Gaussian, Student, and Laplace for which $\gamma$ is larger than $1$ to verify the validity of our approach using Monte Carlo simulations. We observe that that our approximation yields a slightly larger theoretical error than the simulated one, and that the quadratic shrinkage corrects a significant part of this negative bias of estimation, especially around the realized optimal values of $k$. Moreover, we observe that the expected holdout error formula becomes sharper around the optimal $k$ as $\gamma$, the ratio between the fourth moment of the Euclidean norm of the noise vector and the square of the matrix dimension, increases, thereby making the choice of the optimal $k$ more important when running the holdout method in the large but finite dimension setting.

\subsection*{Acknowledgments}
We thank Christian Bongiorno and Damien Challet for the particularly useful discussions and advice that helped improve the article. 

This work was performed using HPC resources from the ``Mésocentre'' computing center of CentraleSup\'elec and \'Ecole Normale Sup\'erieure Paris-Saclay supported by CNRS and R\'egion \^Ile-de-France (\url{http://mesocentre.centralesupelec.fr/}). 

B. C. is supported by JSPS Grant-in-Aid Scientific Research (B) no. 21H00987, and Challenging Research (Exploratory) no.  23K17299.
This research was initiated during the second visit of L. L. to Kyoto University. Both authors acknowledge the hospitality of Kyoto University on that occasion.

\bibliographystyle{unsrt}  
\bibliography{references}  

\begin{thebibliography}{10}

\bibitem{taylor2013putting}
Andy Taylor, Benjamin Joachimi, and Thomas Kitching.
\newblock Putting the precision in precision cosmology: {H}ow accurate should your data covariance matrix be?
\newblock {\em Monthly Notices of the Royal Astronomical Society}, 432(3):1928--1946, 2013.

\bibitem{schlitter1993estimation}
J{\"u}rgen Schlitter.
\newblock Estimation of absolute and relative entropies of macromolecules using the covariance matrix.
\newblock {\em Chemical {P}hysics {L}etters}, 215(6):617--621, 1993.

\bibitem{ottersten1998covariance}
Bj{\"o}rn Ottersten, Peter Stoica, and Richard Roy.
\newblock Covariance matching estimation techniques for array signal processing applications.
\newblock {\em Digital Signal Processing}, 8(3):185--210, 1998.

\bibitem{aubry2017geometric}
Augusto Aubry, Antonio De~Maio, and Luca Pallotta.
\newblock A geometric approach to covariance matrix estimation and its applications to radar problems.
\newblock {\em IEEE Transactions on Signal Processing}, 66(4):907--922, 2017.

\bibitem{varoquaux2010brain}
Ga{\"e}l Varoquaux, Alexandre Gramfort, Jean-Baptiste Poline, and Bertrand Thirion.
\newblock Brain covariance selection: better individual functional connectivity models using population prior.
\newblock {\em Advances in {N}eural {I}nformation {P}rocessing {S}ystems}, 23, 2010.

\bibitem{zhuang2020technical}
Xiaowei Zhuang, Zhengshi Yang, and Dietmar Cordes.
\newblock A technical review of canonical correlation analysis for neuroscience applications.
\newblock {\em Human {B}rain {M}apping}, 41(13):3807--3833, 2020.

\bibitem{laloux1999noise}
Laurent Laloux, Pierre Cizeau, Jean-Philippe Bouchaud, and Marc Potters.
\newblock Noise dressing of financial correlation matrices.
\newblock {\em Physical Review Letters}, 83(7):1467, 1999.

\bibitem{froot1989consistent}
Kenneth~A Froot.
\newblock Consistent covariance matrix estimation with cross-sectional dependence and heteroskedasticity in financial data.
\newblock {\em Journal of Financial and Quantitative Analysis}, 24(3):333--355, 1989.

\bibitem{ledoit2003honey}
O~Ledoit and M~Wolf.
\newblock Honey, {I} shrunk the sample covariance matrix.
\newblock {\em The Journal of Portfolio Management}, 30(4):110--119, 2004.

\bibitem{michaud1989markowitz}
Richard~O Michaud.
\newblock The {M}arkowitz optimization enigma: Is ‘optimized’optimal?
\newblock {\em Financial Analysts Journal}, 45(1):31--42, 1989.

\bibitem{bun2017cleaning}
Jo{\"e}l Bun, Jean-Philippe Bouchaud, and Marc Potters.
\newblock Cleaning large correlation matrices: tools from random matrix theory.
\newblock {\em Physics Reports}, 666:1--109, 2017.

\bibitem{ledoit2011eigenvectors}
Olivier Ledoit and Sandrine P{\'e}ch{\'e}.
\newblock Eigenvectors of some large sample covariance matrix ensembles.
\newblock {\em Probability Theory and Related Fields}, 151(1):233--264, 2011.

\bibitem{ledoit2012nonlinear}
Olivier Ledoit and Michael Wolf.
\newblock Nonlinear shrinkage estimation of large-dimensional covariance matrices.
\newblock {\em Annals of Statistics}, 40:1024--1060, 2012.

\bibitem{ledoit2020analytical}
Olivier Ledoit and Michael Wolf.
\newblock Analytical nonlinear shrinkage of large-dimensional covariance matrices.
\newblock {\em The Annals of Statistics}, 48(5):3043--3065, 2020.

\bibitem{bun2016rotational}
Jo{\"e}l Bun, Romain Allez, Jean-Philippe Bouchaud, and Marc Potters.
\newblock Rotational invariant estimator for general noisy matrices.
\newblock {\em IEEE Transactions on Information Theory}, 62(12):7475--7490, 2016.

\bibitem{bun2018overlaps}
Jo{\"e}l Bun, Jean-Philippe Bouchaud, and Marc Potters.
\newblock Overlaps between eigenvectors of correlated random matrices.
\newblock {\em Physical Review E}, 98(5):052145, 2018.

\bibitem{bartz2016cross}
Daniel Bartz.
\newblock Cross-validation based nonlinear shrinkage.
\newblock {\em Stat}, 1050:2, 2016.

\bibitem{berrar2019cross}
Daniel Berrar.
\newblock Cross-{V}alidation., 2019.

\bibitem{fan2008high}
Jianqing Fan, Yingying Fan, and Jinchi Lv.
\newblock High dimensional covariance matrix estimation using a factor model.
\newblock {\em Journal of Econometrics}, 147(1):186--197, 2008.

\bibitem{yang1994estimation}
Ruoyong Yang and James~O Berger.
\newblock Estimation of a covariance matrix using the reference prior.
\newblock {\em The Annals of Statistics}, pages 1195--1211, 1994.

\bibitem{hijmans2012cross}
Robert~J Hijmans.
\newblock Cross-validation of species distribution models: removing spatial sorting bias and calibration with a null model.
\newblock {\em Ecology}, 93(3):679--688, 2012.

\bibitem{yates2023cross}
Luke~A Yates, Zach Aandahl, Shane~A Richards, and Barry~W Brook.
\newblock Cross validation for model selection: a review with examples from ecology.
\newblock {\em Ecological Monographs}, 93(1):e1557, 2023.

\bibitem{iwama1989phillips}
Naofumi Iwama, Hiromasa Yoshida, Hitoshi Takimoto, Yun Shen, Shuichi Takamura, and Takashige Tsukishima.
\newblock Phillips--{T}ikhonov regularization of plasma image reconstruction with the generalized cross validation.
\newblock {\em Applied Physics Letters}, 54(6):502--504, 1989.

\bibitem{bradshaw2023guide}
Tyler~J Bradshaw, Zachary Huemann, Junjie Hu, and Arman Rahmim.
\newblock A guide to cross-validation for artificial intelligence in medical imaging.
\newblock {\em Radiology: Artificial Intelligence}, 5(4):e220232, 2023.

\bibitem{tan2020estimation}
Vincent Tan and Stefan Zohren.
\newblock Estimation of {L}arge {F}inancial {C}ovariances: {A} {C}ross-{V}alidation {A}pproach.
\newblock {\em arXiv preprint arXiv:2012.05757}, 2020.

\bibitem{Lam2016}
Clifford Lam.
\newblock Nonparametric eigenvalue-regularized precision or covariance matrix estimator.
\newblock {\em The Annals of Statistics}, 44(3):928--953, 2016.

\bibitem{lam2015supplement}
Clifford Lam.
\newblock Supplement to “nonparametric eigenvalue-regularized precision or covariance matrix estimator”, 2015.
\newblock Supplementary material.

\bibitem{engel2017overview}
Jasper Engel, Lutgarde Buydens, and Lionel Blanchet.
\newblock An overview of large-dimensional covariance and precision matrix estimators with applications in chemometrics.
\newblock {\em Journal of chemometrics}, 31(4):e2880, 2017.

\bibitem{tan2025estimation}
Vincent Tan and Stefan Zohren.
\newblock Estimation of large financial covariances: A cross-validation approach.
\newblock {\em Journal of Portfolio Management}, 51(4), 2025.

\bibitem{lam2024holdout}
Lamia Lamrani, Christian Bongiorno, and Marc Potters.
\newblock Optimal data splitting for holdout cross-validation in large covariance matrix estimation.
\newblock {\em arXiv:2503.15186}, 2025.

\bibitem{wick1950evaluation}
Gian-Carlo Wick.
\newblock The evaluation of the collision matrix.
\newblock {\em Physical Review}, 80(2):268, 1950.

\bibitem{collins2017weingarten}
Beno{\^\i}t Collins and Sho Matsumoto.
\newblock Weingarten calculus via orthogonality relations: new applications.
\newblock {\em ALEA, Lat. Am. J. Probab. Math. Stat. \textbf{14}, no. 1, 631--656}, 2017.

\bibitem{weingarten1978asymptotic}
Don Weingarten.
\newblock {A}symptotic behavior of group integrals in the limit of infinite rank.
\newblock {\em J. Math. Phys.(NY)}, 19(5), 1978.

\bibitem{collins2022moment}
Beno{\^\i}t Collins.
\newblock Moment methods on compact groups: {W}eingarten calculus and its applications.
\newblock {\em ICM Vol. IV. Sections 5}, 8:3142--3164, 2022.

\bibitem{saad2013random}
Nadia Abdel Samie Basyouni~Kotb Saad.
\newblock {\em Random matrix theory with applications in statistics and finance}.
\newblock University of Ottawa (Canada), 2013.

\bibitem{collins2014integration}
Beno{\^\i}t Collins, Sho Matsumoto, and Nadia Saad.
\newblock Integration of invariant matrices and moments of inverses of {G}inibre and {W}ishart matrices.
\newblock {\em Journal of Multivariate Analysis}, 126:1--13, 2014.

\bibitem{potters2020first}
Marc Potters and Jean-Philippe Bouchaud.
\newblock {\em A first course in random matrix theory: for physicists, engineers and data scientists}.
\newblock Cambridge University Press, 2020.

\bibitem{collins2009some}
Beno{\^\i}t Collins and Sho Matsumoto.
\newblock On some properties of orthogonal weingarten functions.
\newblock {\em Journal of Mathematical Physics}, 50(11), 2009.

\bibitem{haff1979identity}
LR~Haff.
\newblock An identity for the {W}ishart distribution with applications.
\newblock {\em Journal of Multivariate Analysis}, 9(4):531--544, 1979.

\bibitem{ledoit2022quadratic}
Olivier Ledoit and Michael Wolf.
\newblock Quadratic shrinkage for large covariance matrices.
\newblock {\em Bernoulli}, 28(3):1519--1547, 2022.

\bibitem{kendall1948advanced}
Maurice~George Kendall.
\newblock {\em The advanced theory of statistics. Vols. 1.}
\newblock Oxford University Press, 1948.

\end{thebibliography}
\newpage
\section{Appendix}
\label{section:appendix}

\subsection{Computation of $\mathbb{E}\left[\tau(\bm{E}^{3}) \right]$ and $\mathbb{E}\left[\tau(\bm{E}^{4}) \right]$}

List of moments for $\bm{\Sigma}$ an inverse Wishart matrix of parameters $n$ and $p$ and its sample covariance matrix $\bm{E}$ derived from a rotationally invariant multiplicative noise
\begin{itemize}
    \item $\mathbb{E}\left[\tau(\bm{E})\right]=\mathbb{E}\left[\tau(\bm{\Sigma})\right]=1$
    \item $\mathbb{E}\left[\tau(\bm{\Sigma}^{2})\right]=\mathbb{E}\left[\tau(\bm{E\Sigma})\right]=1+p$
    \item $\mathbb{E}\left[\tau(\bm{E}^{2})\right]=1+p+q\gamma$ with $\gamma=\frac{\mathbb{E}\left[\Tr((\bm{\xi\xi}^{T})^{2}) \right]}{n(n+2)}$.
    \item 
    
    $\mathbb{E}\left[\tau(\bm{E}^{2}\bm{\Sigma}) \right]= \mathbb{E}\left[\tau(\bm{\Sigma}^{3})\right] + q\gamma\mathbb{E}\left[\tau(\bm{\Sigma}^{2}) \right]$
\end{itemize}

\begin{lemma}
\label{lemma_sample_cov_moments}
Let $\bm{\xi}\in\mathbb{R}^{n}$ be a rotationally invariant noise vector such that
\begin{equation}
    \mathbb{E}\left[\bm{\xi\xi}^{T} \right]=\bm{\mathds{1}},
\end{equation}
and let $\bm{E}$ the sample covariance of $\bm{\Sigma}$ generated over $t$ obervations of the noise $\bm{\xi}$. Then in the high dimension limit
\begin{equation}
        \mathbb{E}\left[\tau(\bm{E}^{3}) \right]= \mathbb{E}\left[\tau(\bm{\Sigma}^{3}) \right] +q\frac{\mathbb{E}\left[ \Tr((\bm{\xi\xi}^{T})^{2})\right]}{n(n+2)}\mathbb{E}\left[3\tau(\bm{\Sigma}^{2})\tau(\bm{\Sigma}) \right] + q^{2}\frac{\mathbb{E}\left[\Tr ((\bm{\xi}\bm{\xi}^{T})^{3}) \right]}{n(n+2)(n+4)}\mathbb{E}\left[\tau(\bm{\Sigma)^{3}} \right]
    \end{equation}
    and
\begin{align}
    \mathbb{E}\left[\tau(\bm{E}^{4}) \right]&=\mathbb{E}\left[\tau(\bm{\Sigma}^{4}) \right] + q\frac{\mathbb{E}\left[ \Tr((\bm{\xi}\bm{\xi}^{T})^{2})\right]}{n(n+2)}\mathbb{E}\left[4\tau(\bm{\Sigma}^{3})\tau(\bm{\Sigma}) + 2\tau(\bm{\Sigma})^{2} \right] \notag
\\
    &+ q^{2}\frac{\mathbb{E}\left[\Tr((\bm{\xi\xi}^{T})^{3}) \right]}{n(n+2)(n+4)}\mathbb{E}\left[4\tau(\bm{\Sigma}^{2})\tau(\bm{\Sigma}^{2}) \right] + q^{3}\frac{\mathbb{E}\left[\Tr((\bm{\xi\xi}^{T})^{4})\right]}{n(n+2)(n+4)(n+6)}\mathbb{E}\left[ (\tau(\bm{\Sigma}))^{4}\right].
    \end{align}    

Let us assume that $\bm{\Sigma}$ is an inverse Wishart of parameters $n$ and $p$ such that $p$ is negligible in front of $n$, then in the limit of high dimension the third and fourth order moments of the sample covariance matrix equal\begin{equation}
        \mathbb{E}\left[\tau(\bm{E}^{3}) \right]= 1+3p+2p^{2}+q\frac{\mathbb{E}\left[ \Tr((\bm{\xi\xi}^{T})^{2})\right]}{n(n+2)}(3(1+p)) + q^{2}\frac{\mathbb{E}\left[\Tr ((\bm{\xi}\bm{\xi}^{T})^{3}) \right]}{n(n+2)(n+4)}
    \end{equation}
    and

    \begin{align}
    \mathbb{E}\left[\tau(\bm{E}^{4}) \right]&=1+6p+11p^{2}+5p^{3} + 2q(3+6p+4p^{2})\frac{\mathbb{E}\left[ \Tr((\bm{\xi}\bm{\xi}^{T})^{2})\right]}{n(n+2)}\notag
 \\
    &+ 4q^{2}(1+p)^{2}\frac{\mathbb{E}\left[\Tr((\bm{\xi\xi}^{T})^{3}) \right]}{n(n+2)(n+4)} + q^{3}\frac{\mathbb{E}\left[\Tr((\bm{\xi\xi}^{T})^{4})\right]}{n(n+2)(n+4)(n+6)}.
    \end{align}    
\end{lemma}

\begin{proof}

Let us first explain the main argument of the proof. Since $\bm{\xi}\in\mathbb{R}^{n}$ is rotationally invariant vector, it can be decomposed as 
\begin{equation}
    \bm{\xi}=s\bm{u},
\end{equation}
where $s$ is a random variable independent from $\bm{u}$ a uniform random vector on the sphere $\mathcal{S}^{n-1}$. Let us now consider $i=(i_{1},...,i_{2k})$ a sequence of integers between $1$ and $n$, then using the independence between $s$ and $\bm{u}$ we obtain

\begin{equation}
\label{eq:prod_components_ri}
\mathbb{E}\left[\xi_{i_{1}}\xi_{i_{2}}...\xi_{2k} \right]=\mathbb{E}\left[s^{2k} \right]\mathbb{E}\left[u_{i_{1}} u_{i_{2}}...u_{i_{2k}}\right].
\end{equation}

Then the difficult question is the computation of $\mathbb{E}\left[u_{i_{1}} u_{i_{2}}...u_{i_{2k}}\right]$. To compute this quantity, let us consider $\bm{y}\in\mathbb{R}^{n}$ a centered Gaussian random vector with an identity covariance matrix. $\bm{y}$ is rotationally invariant and can therefore be written $\bm{y}=s_{G}\bm{v}$ with $s_{G}$ a random variable independent from $\bm{v}$ a uniform random vector on the sphere $\mathcal{S}^{n-1}$.

Therefore Eq.\ \eqref{eq:prod_components_ri} can be rewritten as
\begin{equation}
\label{eq:argument_gaussian}
\mathbb{E}\left[\xi_{i_{1}}\xi_{i_{2}}...\xi_{i_{2k}} \right]=\frac{\mathbb{E}\left[s^{2k} \right]}{\mathbb{E}\left[s_{G}^{2k}\right]}\mathbb{E}\left[y_{i_{1}} y_{i_{2}}...y_{i_{2k}}\right] 
\end{equation}
where the term $\mathbb{E}\left[ y_{i_{1}} y_{i_{2}}...y_{i_{2k}}\right]$ can be computed using Wick theorem \cite{wick1950evaluation}.

Therefore, to compute the thirs and fourth order moments of $\bm{E}$, i.e. $\mathbb{E}\left[\tau(\bm{\Sigma}^{3})\right]$ and $\mathbb{E}\left[\tau(\bm{\Sigma}^{4})\right]$, the first step is to write them as a function of moments of elements $\bm{\xi}$ using the i.i.d. assumption on the columns of the noise matrix $\bm{X}$.

\begin{align}
    \mathbb{E}\left[\tau(\bm{E}^{3})\right]&= \frac{1}{nt^{3}}\sum \mathbb{E}\left[ \Tilde{\Sigma}_{ii_{1}}X_{i_{1}t_{1}}X_{i_{2}t_{1}}\Tilde{\bm{\Sigma}}_{i_{2}j}\Tilde{\bm{\Sigma}}_{ji_{3}}X_{i_{3}t_{2}}X_{i_{4}t_{2}}\Tilde{\bm{\Sigma}}_{i_{4}k} \Tilde{\bm{\Sigma}}_{ki_{5}}X_{i_{5}t_{3}}X_{i_{6}t_{3}}\Tilde{\bm{\Sigma}}_{i_{6}i} \right]\notag \\
    &=\frac{t(t-1)(t-2)}{nt^{3}}\mathbb{E}(\tau(\bm{\Sigma}^{3}))\notag \\
    &+\frac{t(t-1)}{nt^{3}}\sum\mathbb{E}\left[\Sigma_{ij} \Tilde{\bm{\Sigma}}_{ji_{3}}X_{i_{3}1}X_{i_{4}1}\Tilde{\bm{\Sigma}}_{i_{4}k} \Tilde{\bm{\Sigma}}_{ki_{5}}X_{i_{5}1}X_{i_{6}}\Tilde{\bm{\Sigma}}_{i_{6}1} \right]\notag \\
    &+\frac{t(t-1)}{nt^{3}}\mathbb{E}\left[ \Tilde{\Sigma}_{ii_{1}}X_{i_{1}t_{1}}X_{i_{2}t_{1}}\Tilde{\bm{\Sigma}}_{i_{2}j}\Sigma_{jk} \Tilde{\bm{\Sigma}}_{ki_{5}}X_{i_{5}t_{3}}X_{i_{6}t_{3}}\Tilde{\bm{\Sigma}}_{i_{6}i} \right]\notag \\
     &+\frac{t(t-1)}{nt^{3}}\mathbb{E}\left[ \Tilde{\Sigma}_{ii_{1}}X_{i_{1}t_{1}}X_{i_{2}t_{1}}\Tilde{\bm{\Sigma}}_{i_{2}j}\Tilde{\bm{\Sigma}}_{ji_{3}}X_{i_{3}t_{2}}X_{i_{4}t_{2}}\Tilde{\bm{\Sigma}}_{i_{4}k} \Sigma_{ki} \right]\notag \\
    &+\frac{1}{nt^{2}}\sum \mathbb{E}\left[ \Tilde{\bm{\Sigma}}_{ii_{1}}X_{i_{1}1}X_{i_{2}1}\Tilde{\bm{\Sigma}}_{i_{2}j}\Tilde{\bm{\Sigma}}_{ji_{3}}X_{i_{3}1}X_{i_{4}1}\Tilde{\bm{\Sigma}}_{i_{4}k} \Tilde{\bm{\Sigma}}_{ki_{5}}X_{i_{5}1}X_{i_{6}}\Tilde{\bm{\Sigma}}_{i_{6}i} \right]
\end{align}

By writing the previous equation as a function of $\bm{\xi}$ we obtain

\begin{align}
\label{eq:3rd_moment_before_wick}
\mathbb{E}\left[\tau(\bm{E}^{3}) \right]&=\frac{t(t-1)(t-2)}{t^{3}}\mathbb{E}\left[\tau(\bm{\Sigma}^{3}) \right] \notag \\
&+\frac{t(t-1)}{t^{3}}\sum_{i,j,k,i_{3},i_{4},i_{5},i_{6}=1}^{n}\mathbb{E}\left[\Sigma_{ij} \Tilde{\bm{\Sigma}}_{ji_{3}}\Tilde{\bm{\Sigma}}_{i_{4}k} \Tilde{\bm{\Sigma}}_{ki_{5}}\Tilde{\bm{\Sigma}}_{i_{6}1}\right]\mathbb{E}\left[\xi_{i_{3}} \xi_{i_{4}}\xi_{i_{5}}\xi_{i_{6}}\right]\notag\\
&+\frac{t(t-1)}{nt^{3}}\sum_{i,j,k,i_{1},i_{2},i_{5},i_{6}=1}^{n}\mathbb{E}\left[\Tilde{\Sigma}_{ii_{1}} \Tilde{\bm{\Sigma}}_{i_{2}j}\Sigma_{jk} \Tilde{\bm{\Sigma}}_{ki_{5}}\Tilde{\bm{\Sigma}}_{i_{6}i} \right]\mathbb{E}\left[ \xi_{i_{1}}\xi_{i_{2}}\xi_{i_{5}}\xi_{i_{6}}\right]\notag\\ 
&+\frac{t(t-1)}{nt^{3}}\sum_{i,j,k,i_{1},i_{2},i_{3},i_{4}=1}^{n} \mathbb{E}\left[ \Tilde{\Sigma}_{ii_{1}}\Tilde{\bm{\Sigma}}_{i_{2}j}\Tilde{\bm{\Sigma}}_{ji_{3}}\Tilde{\bm{\Sigma}}_{i_{4}k} \Sigma_{ki} \right] \mathbb{E}\left[ \xi_{i_{1}}\xi_{i_{2}}\xi_{i_{3}}\xi_{i_{4}}\right]\notag \\
&=\frac{1}{nt^{2}}\sum_{i,j,k}\mathbb{E}\left[\Tilde{\Sigma}_{ii_{1}}\Tilde{\Sigma}_{i_{2}j}\Tilde{\Sigma}_{ji_{3}}\Tilde{\Sigma}_{i_{4}k}\Tilde{\Sigma}_{ki_{5}}\Tilde{\Sigma}_{i_{6}i} \right]\mathbb{E}\left[ \xi_{i_{1}}\xi_{i_{2}}\xi_{i_{3}}\xi_{i_{4}}\xi_{i_{5}}\xi_{i_{6}}\right]
\end{align}

Then we link Eq.\ \eqref{eq:3rd_moment_before_wick} with Gaussian moments using the argument written in Eq.\ \eqref{eq:argument_gaussian} and after applying Wick theorem, we obtain 

\begin{align}
\label{eq:3rd_moment_before_hd}
\mathbb{E}\left[\tau(\bm{E}^{3}) \right]&=\frac{t(t-1)(t-2)}{t^{3}}\mathbb{E}\left[\tau(\bm{\Sigma}^{3}) \right]\notag
\\ 
&+\frac{3t(t-1)}{t^{3}}\frac{\mathbb{E}\left[\Tr((\bm{\xi}\bm{\xi}^{T})^{2} \right]}{n(n+2)} \mathbb{E}\left[\Tr(\bm{\Sigma}^{3})+ \Tr(\bm{\Sigma}^{3}) + \Tr(\bm{\Sigma}^{2})\Tr(\bm{\Sigma})\right]\notag\\
&=\frac{1}{nt^{2}}\frac{\mathbb{E}\left[\Tr((\bm{\xi}\bm{\xi}^{T})^{3} \right]}{n(n+2)(n+4)}\mathbb{E}\left[ 8\Tr(\bm{\Sigma}^{3})+6\Tr(\bm{\Sigma}^{2})\Tr(\bm{\Sigma}) + \Tr(\bm{\Sigma})^{3}\right], 
\end{align}
and finally by taking the high dimensional limit in Eq.\ \eqref{eq:3rd_moment_before_hd}, i.e. $n,t\to\infty$ with $\frac{n}{t}=q>0$, we finally obtain

\begin{equation}
\label{eq:3rd_moment_after_hd}    \mathbb{E}\left[\tau(\bm{E}^{3})\right]=\mathbb{E}\left[\tau(\bm{\Sigma}^{3})\right]+3q\frac{\mathbb{E}\left[\Tr((\bm{\xi}\bm{\xi}^{T})^{3})\right]}{n(n+2)}\mathbb{E}\left[ \Tr(\bm{\Sigma}^{2})\Tr(\bm{\Sigma})\right] + q^{2}\frac{\mathbb{E}\left[\Tr(\bm{\xi}\bm{\xi}^{T})^{3} \right]}{n(n+2)(n+4)}\mathbb{E}\left[\Tr(\bm{\Sigma}^{3})\right]
\end{equation}

Now let us focus on the fourth order moment
\begin{equation}
\label{eq:third_order_moment_E}
    \mathbb{E}\left[\tau(\bm{E}^{4}) \right]=\frac{1}{nt^{4}}\sum_{\substack{i,j,k,l, \\ i_{1},\dots,i_{8}=1}}^{n} \sum_{t_{1},...,t_{4}=1}^{t}\mathbb{E}\left[\Tilde{\Sigma}_{ii_{1}}X_{i_{1}t_{1}}X_{i_{2}t_{1}}\Tilde{\Sigma}_{i_{2}j}\Tilde{\Sigma}_{ji_{3}}X_{i_{3}t_{2}}X_{i_{4}t_{2}}\Tilde{\Sigma}_{i_{4}k}\Tilde{\Sigma}_{ki_{5}}X_{i_{5}t_{3}}X_{i_{6}t_{3}}\Tilde{\Sigma}_{i_{6}l}\Tilde{\Sigma}_{li_{7}}X_{i_{7}t_{4}}X_{i_{8}t_{4}}\Tilde{\Sigma}_{i_{8}i} \right].
\end{equation}

The first step is to write the previous equation as a function of columns of the noise vector $\bm{\xi}$ using the i.i.d. assumption on the columns of $\bm{X}$ and the independence between $\bm{\Sigma}$ and $\bm{X}$, we obtain

\begin{align}
\label{eq:4th_moment_before_wick}
\mathbb{E}\left[\tau(\bm{E}^{4}) \right] &=\frac{t(t-1)(t-2)(t-3)}{nt^{4}}\mathbb{E}\left[\tau(\bm{\Sigma^{4}}) \right] \notag \\
&+\frac{t(t-1)(t-2)}{nt^{4}}\sum_{i_{1},i_{2},i_{3},i_{4}=1}^{n}\mathbb{E}\left[\Tilde{\Sigma}_{ii_{1}}\Tilde{\Sigma}_{i_{2}j}\Tilde{\Sigma}_{ji_{3}}\Tilde{\Sigma}_{i_{4}k}\Sigma_{kl}\Sigma_{li} \right]\mathbb{E}\left[\xi_{i_{1}} \xi_{i_{2}}\xi_{i_{3}}\xi_{i_{4}}\right] \notag \\
&+\frac{t(t-1)(t-2)}{nt^{4}}\sum_{i_{1},i_{2},i_{5},i_{6}=1}^{n}\mathbb{E}\left[\Tilde{\Sigma}_{ii_{1}}\Tilde{\Sigma}_{i_{2}j}\Sigma_{jk}\Tilde{\Sigma}_{ki_{5}}\Tilde{\Sigma}_{i_{6}l}\Sigma_{li}\right]\mathbb{E}\left[\xi_{i_{1}} \xi_{i_{2}}\xi_{i_{5}}\xi_{i_{6}}\right] \notag \\
&+\frac{t(t-1)(t-2)}{nt^{4}}\sum_{i_{1},i_{2},i_{7},i_{8}=1}^{n}\mathbb{E}\left[\Tilde{\Sigma}_{ii_{1}}\Tilde{\Sigma}_{i_{2}j}\Sigma_{jk}\Sigma_{kl}\Tilde{\Sigma}_{li_{7}}\Tilde{\Sigma}_{i_{8}i}\right]\mathbb{E}\left[\xi_{i_{1}} \xi_{i_{2}}\xi_{i_{7}}\xi_{i_{8}}\right] \notag \\
&+\frac{t(t-1)(t-2)}{nt^{4}}\sum_{i_{3},i_{4},i_{5},i_{6}=1}^{n}\mathbb{E}\left[\Sigma_{ij}\Tilde{\Sigma}_{ji_{3}}\Tilde{\Sigma}_{i_{4}k}\Tilde{\Sigma}_{ki_{5}}\Tilde{\Sigma}_{i_{6}l}\Sigma_{li}\right]\mathbb{E}\left[\xi_{i_{3}} \xi_{i_{4}}\xi_{i_{5}}\xi_{i_{6}}\right] \notag \\
&+\frac{t(t-1)(t-2)}{nt^{4}}\sum_{i_{3},i_{4},i_{7},i_{8}=1}^{n}\mathbb{E}\left[\Sigma_{ij}\Tilde{\Sigma}_{ji_{3}}\Tilde{\Sigma}_{i_{4}k}\Sigma_{kl}\Tilde{\Sigma}_{li_{7}}\Tilde{\Sigma}_{i_{8}i}\right]\mathbb{E}\left[\xi_{i_{3}} \xi_{i_{4}}\xi_{i_{7}}\xi_{i_{8}}\right] \notag \\
&+\frac{t(t-1)(t-2)}{nt^{4}}\sum_{i_{5},i_{6},i_{7},i_{8}=1}^{n}\mathbb{E}\left[\Sigma_{ij}\Sigma_{jk}\Tilde{\Sigma}_{ki_{5}}\Tilde{\Sigma}_{i_{6}l}\Tilde{\Sigma}_{li_{7}}\Tilde{\Sigma}_{i_{8}i}\right]\mathbb{E}\left[\xi_{i_{5}} \xi_{i_{6}}\xi_{i_{7}}\xi_{i_{8}}\right] \notag \\
&+\frac{t(t-1)}{nt^{4}}\sum_{i_{3},i_{4},i_{5},i_{6},i_{7},i_{8}=1}^{n}
\mathbb{E}\left[ \Sigma_{ij} \Tilde{\Sigma}_{ji_{3}} \Tilde{\Sigma}_{i_{4}k} \Tilde{\Sigma}_{ki_{5}}\Tilde{\Sigma}_{i_{6}l}\Tilde{\Sigma}_{li_{7}}\Tilde{\Sigma}_{i_{8}i}\right]\mathbb{E}\left[ \xi_{i_{3}}\xi_{i_{4}}\xi_{i_{5}}\xi_{i_{6}}\xi_{i_{7}}\xi_{i_{8}}\right]\notag \\
&+\frac{t(t-1)}{nt^{4}}\sum_{i_{1},i_{2},i_{5},i_{6},i_{7},i_{8}=1}^{n}
\mathbb{E}\left[\Tilde{\Sigma}_{ii_{1}} \Tilde{\Sigma}_{i_{2}j}\Sigma_{jk} \Tilde{\Sigma}_{ki_{5}}\Tilde{\Sigma}_{i_{6}l}\Tilde{\Sigma}_{li_{7}}\Tilde{\Sigma}_{i_{8}i}\right]\mathbb{E}\left[ \xi_{i_{1}}\xi_{i_{2}}\xi_{i_{5}}\xi_{i_{6}}\xi_{i_{7}}\xi_{i_{8}}\right]\notag \\
&+\frac{t(t-1)}{nt^{4}}\sum_{i_{1},i_{2},i_{3},i_{4},i_{7},i_{8}=1}^{n}
\mathbb{E}\left[\Tilde{\Sigma}_{ii_{1}} \Tilde{\Sigma}_{i_{2}j}\Tilde{\Sigma}_{ji_{3}}\Tilde{\Sigma}_{i_{4}k}\Sigma_{kl}\Tilde{\Sigma}_{li_{7}}\Tilde{\Sigma}_{i_{8}i}\right]\mathbb{E}\left[ \xi_{i_{1}}\xi_{i_{2}}\xi_{i_{3}}\xi_{i_{4}}\xi_{i_{7}}\xi_{i_{8}}\right]\notag \\
&+\frac{t(t-1)}{nt^{4}}\sum_{i_{1},i_{2},i_{3},i_{4},i_{7},i_{8}=1}^{n}
\mathbb{E}\left[\Tilde{\Sigma}_{ii_{1}} \Tilde{\Sigma}_{i_{2}j}\Tilde{\Sigma}_{ji_{3}}\Tilde{\Sigma}_{i_{4}k}\Tilde{\Sigma}_{ki_{5}}\Tilde{\Sigma}_{i_{6}l}\Sigma_{li}\right]\mathbb{E}\left[ \xi_{i_{1}}\xi_{i_{2}}\xi_{i_{3}}\xi_{i_{4}}\xi_{i_{5}}\xi_{i_{6}}\right]\notag \\
&+\frac{1}{nt^{3}}\sum_{i_{1},i_{2},i_{3},i_{4},i_{5},i_{6},i_{7},i_{8}=1}^{n}\mathbb{E}\left[\Tilde{\Sigma}_{ii_{1}}\Tilde{\Sigma}_{i_{2}j}\Tilde{\Sigma}_{ji_{3}}\Tilde{\Sigma}_{i_{4}k}\Tilde{\Sigma}_{ki_{5}}\Tilde{\Sigma}_{i_{6}l}\Tilde{\Sigma}_{li_{7}}\Tilde{\Sigma}_{i_{8}i} \right] \mathbb{E}\left[\xi_{i_{1}} \xi_{i_{2}}\xi_{i_{3}}\xi_{i_{4}}\xi_{i_{5}}\xi_{i_{6}}\xi_{i_{7}}\xi_{i_{8}}       \right]
\end{align}

Then we link Eq.\ \eqref{eq:4th_moment_before_wick} with Gaussian moments using the argument written in Eq.\ \eqref{eq:argument_gaussian} and after applying Wick theorem, we obtain 

\begin{align}
\mathbb{E}\left[\tau(\bm{E}^{4}) \right] &= \frac{t(t-1)(t-2)(t-3)}{t^{4}}\mathbb{E}\left[ \tau(\bm{\Sigma}^{4})\right] \notag\\
 &+\frac{t(t-1)(t-2)}{nt^{4}}\frac{\mathbb{E}\left[\Tr((\bm{\xi}\bm{\xi}^{T})^{2})\right]}{n(n+2)}\mathbb{E}\left[12\Tr(\bm{\Sigma}^{4})+4\Tr(\bm{\Sigma}^{3})\Tr(\bm{\Sigma}) +2\Tr(\bm{\Sigma}^{2})^{2}\right]\notag \\
 &+\frac{t(t-1)}{nt^{4}}\frac{\mathbb{E}\left[\Tr((\bm{\xi}\bm{\xi}^{T})^{3})\right]}{n(n+2)(n+4)}\mathbb{E}\left[32 \Tr(\bm{\Sigma}^{4}) + 16\Tr(\bm{\Sigma}^{3})\Tr(\bm{\Sigma})+ 8\Tr(\bm{\Sigma}^{2})^{2} + 4\Tr(\bm{\Sigma}^{2})\Tr(\bm{\Sigma})^{2}\right]\notag\\
 &+\frac{1}{nt^{3}}\frac{\mathbb{E}\left[\Tr((\bm{\xi}\bm{\xi}^{T})^{4} \right]}{n(n+2)(n+4)(n+6)}\mathbb{E}\left[ 48\Tr(\bm{\Sigma}^{4}) + 32\Tr(\bm{\Sigma}^{3})\Tr(\bm{\Sigma})+ 14\Tr(\bm{\Sigma}^{2})^{2} + 13\Tr(\bm{\Sigma}^{2})\Tr(\bm{\Sigma})^{2} + \Tr(\bm{\Sigma})^{4}\right],
\end{align}

and be taking the high dimensional limit for $n,t\to\infty$ with $\frac{n}{t}\to q$, we obtain

\begin{align}
\label{eq:fourth_order_E}
 \mathbb{E}\left[\tau(\bm{E}^{4}) \right] &=\mathbb{E}\left[ \tau(\bm{\Sigma}^{4})\right]+q\frac{\mathbb{E}\left[\Tr((\bm{\xi}\bm{\xi}^{T})^{2})\right]}{n(n+2)}\mathbb{E}\left[4\tau(\bm{\Sigma}^{3})\tau(\bm{\Sigma}) +2\tau(\bm{\Sigma}^{2})^{2}\right]\notag\\
 &+q^{2}\frac{\mathbb{E}\left[\Tr((\bm{\xi}\bm{\xi}^{T})^{3})\right]}{n(n+2)(n+4)}\mathbb{E}\left[4\tau(\bm{\Sigma}^{2})\tau(\bm{\Sigma})^{2}\right]+q^{3}\frac{\mathbb{E}\left[\Tr((\bm{\xi}\bm{\xi}^{T})^{4} \right]}{n(n+2)(n+4)(n+6)}\mathbb{E}\left[\tau(\bm{\Sigma})^{4}\right]
\end{align}

Then to obtain the third and fourth order tracial moments of $\bm{E}$ when $\bm{\Sigma}$ is an inverse Wishart with parameters $n$ and $p$, we recover the cumulants of $\bm{\Sigma}$ by computing the Taylor expansion of the R-transform of $\bm{\Sigma}$ around $0$. Indeed, the R-transform of $\bm{\Sigma}$ equals \cite{potters2020first}
\begin{equation}
    R_{\bm{\Sigma}}(x)=\frac{1-\sqrt{1-4px}}{2px},
\end{equation}
which gives the following Taylor expansion around $0$ up to order $3$
\begin{equation}
    R_{\bm{\Sigma}}(x)\underset{x \to 0}{=}1+px+2p^{2}x^{2}+5p^{3}x^{3}+o(x^3)
\end{equation}

By identification, the four first cumulants are equal to

 \[
\begin{cases}
\kappa_{1}=1 \\
\kappa_{2}=p \\
\kappa_{3}=2p^{2} \\
\kappa_{4}=5p^{3}. 
\end{cases}
\]

Then using the recursive formula linking moments and cumulants \cite{kendall1948advanced}, we obtain 

\begin{equation}
\label{eq:list_moments_sigma_invW}
\begin{array}{r@{\;}l@{\;}l}
\left\{
\begin{array}{l}
\tau(\bm{\Sigma}) = \kappa_{1} \\
\tau(\bm{\Sigma}^{2}) = \kappa_{2} - \kappa_{1}^{2} \\
\tau(\bm{\Sigma}^{3}) = \kappa_{3} + 3\kappa_{2}\kappa_{1} + \kappa_{1}^{3} \\
\tau(\bm{\Sigma}^{4}) = \kappa_{4} + 4\kappa_{1}\kappa_{3} + 3\kappa_{2}^{2} + 6\kappa_{2}\kappa_{1}^{2} + \kappa_{1}^{4}
\end{array}\right.
& \implies \hspace{5mm} &
\left\{
\begin{array}{l}
\tau(\bm{\Sigma}) = 1 \\
\tau(\bm{\Sigma}^{2}) = 1+p\\
\tau(\bm{\Sigma}^{3}) = 1+3p+2p^{2} \\
\tau(\bm{\Sigma}^{4}) = 1+6p+11p^{2}+5p^{3}
\end{array}\right.
\end{array} 
\end{equation}

and by re-injecting the list of moments of Eq.\ \eqref{eq:list_moments_sigma_invW} in Eq.\ \eqref{eq:third_order_moment_E} and Eq.\ \eqref{eq:fourth_order_E}, we obtain the desired result.

\end{proof}
\end{document}